\documentclass[journal,10pt]{IEEEtran}
\usepackage{mathrsfs}
\usepackage{bbding}
\usepackage{amsfonts}
\usepackage{amsmath}
\usepackage{graphicx}
\usepackage{subfigure}
\usepackage{latexsym}
\usepackage{amssymb}
\usepackage{amsmath}
\usepackage{enumerate}
\usepackage{color}
\usepackage[font=small]{caption}

\newtheorem{Assumption}{Assumption}
\newtheorem{Algorithm}{Algorithm}
\newtheorem{Definition}{Definition}
\newtheorem{Lemma}{Lemma}
\newtheorem{Problem}{Problem}
\newtheorem{Remark}{Remark}
\newtheorem{Theorem}{Theorem}
\newtheorem{Corollary}{Corollary}
\newtheorem{Claim}{Claim}
\newcommand{\tcb}{\textcolor{black}}

\allowdisplaybreaks[4]

\begin{document}
\title{Low Complexity Delay-Constrained Beamforming for  Multi-User MIMO  Systems   with Imperfect CSIT}

\author
{\IEEEauthorblockN{Vincent K. N. Lau, \emph{FIEEE}, Fan Zhang, \emph{StMIEEE}, Ying Cui, \emph{MIEEE}} \thanks{Vincent K. N. Lau and Fan Zhang are with Department of ECE, Hong Kong University of Science and Technology, Hong Kong. Ying Cui is with Department of ECE, Northeastern University, USA.}
}
\maketitle

\begin{abstract} 

In this paper, we consider the delay-constrained beamforming control for  downlink multi-user MIMO (MU-MIMO) systems with imperfect channel state information at the transmitter (CSIT). The delay-constrained control problem is formulated as an infinite horizon average cost partially observed Markov decision process.  To deal with the  curse of dimensionality, we introduce a virtual continuous time system and derive a closed-form approximate  value function  using perturbation analysis w.r.t. the CSIT errors. To deal with the challenge of  the conditional packet error rate (PER), we build a tractable closed-form approximation using a Bernstein-type  inequality. Based on the closed-form approximations of the relative value function and the conditional PER, we propose a conservative formulation of the original beamforming control problem. The conservative problem is non-convex and we transform it into a convex problem using  the semidefinite relaxation (SDR) technique. We then propose an alternating iterative algorithm to solve the SDR problem. Finally, the proposed scheme is compared  with various baselines through simulations and it is shown that significant  performance gain can be achieved.

\end{abstract}

\section{introduction}
There have has intense research interest in using multiple antenna technology to boost the capacity of wireless systems \cite{MIMOBC2}, \cite{mmse}. In \cite{mmse}, \cite{precoder1},  the authors show that substantial capacity gain can be achieved in  downlink multi-user MIMO (MU-MIMO) systems using simple zero-forcing (ZF) or minimum mean square (MMSE) precoders when the channel state  information (CSI) at the transmitter  (CSIT) is perfect.  However, the CSIT measured at the BS cannot be perfect due to either the CSIT estimation noise or the outdatedness of the CSIT resulting from duplexing delay.   In \cite{sumgood2}--\cite{Robustprob1}, the authors consider  robust beamforming design to maximize the sum goodput  \cite{sumgood2}, minimize the MMSE \cite{mmse2}  or  the transmit power \cite{Robustprob1} of the downlink MU-MIMO system subject to either the worst-case  SINR  constraints \cite{mmse2}, \cite{Robustprob1} or the probabilistic SINR constraints \cite{sumgood2}. To simplify  the associated optimization problem, semidefinite relaxation (SDR) technique  and majorization theory  are applied in \cite{luo}--\cite{major2}.  However, all these works focus on physical layer performance and  ignore the bursty data   arrivals as well as the delay requirement of information flows.  The resulting control policy is adaptive to the CSI/CSIT only and cannot guarantee good delay performance for delay-sensitive applications. In general, physical layer oriented designs cannot guarantee good delay performance  \cite{surveydelay}.   The delay-aware control policy  should be adaptive to both the CSI and the queue state information (QSI). This is because  the CSI provides information regarding the channel opportunity  while the QSI indicates the urgency of the data flows.

The design framework taking into account the queuing delay performance of information flows is highly non-trivial as it involves both queuing theory (to model the queuing dynamics) and information theory (to model the physical layer dynamics). The control policy will affect the underlying probability measure (or stochastic evolution) of the system state (CSI, QSI) and the  state process evolves stochastically as a \emph{controlled Markov chain} for a given  policy. A systematic approach  to solve the stochastic optimization problem  is through  Markov Decision Process (MDP) \cite{Cao}, \cite{DP_Bertsekas}. In general,  the optimal control policy can be obtained by solving the well-known \emph{Bellman equation} using numerical methods such as brute-force value iteration and policy iteration \cite{DP_Bertsekas}. However, this usually cannot lead to any desirable solutions because solving the Bellman equation involves solving an exponentially large system of non-linear equations, which induces huge complexity  (i.e., the curse of dimensionality). \tcb{There are some existing works that use \emph{stochastic approximation approach} to  deal with the complexity  issue \cite{downmimo1}, \cite{downmimo2}. Specifically, the value function is approximated by the sum of the per-flow functions. The per-flow functions are then estimated  using distributed online learning  algorithms, which have linear complexity. However, the stochastic learning approach can only give numerical solution to the Bellman equation and may suffer from slow convergence and lack of insight. }

In this paper, we consider a downlink MU-MIMO system with imperfect CSIT, where  a multi-antenna BS communicates to $K$ single-antenna mobiles. We focus on minimizing the average power of the BS  subject to the average delay constraints of the $K$ bursty data flows.   There are several first order technical challenges associated with the stochastic optimization problem due to the imperfect CSIT and the average delay constraints.
\begin{itemize}
	\item	\textbf{Challenges due to the Mutual Coupling  of the $K$ Queues:}
	Multi-user interference in the downlink MU-MIMO system cannot be completely eliminated under the  imperfect CSIT. As the service rate of the $k$-th queue depends on the transmit power of the other mobiles via  interference, the queue dynamics of the $K$ mobiles in the system are mutually coupled together.  Therefore, the associated stochastic optimization problem is a $K$-dimensional MDP \cite{surveydelay}. There will be the curse of dimensionality issue while solving the associated Bellman equation and  standard MDP solutions  have exponential complexity in $K$. 
		
	\item \textbf{Challenge due to the Packet Error Probability:}  
	The imperfect CSIT leads to systematic packet errors due to  channel outage\footnote{Under imperfect CSIT, systematic packet errors occurs
whenever the scheduled data rate exceeds the instantaneous mutual information (namely, channel outage) despite the use of powerful error correction coding.}. Therefore,  it is important to consider the packet error rate (PER) in the optimization. However, this involves obtaining the conditional probability distribution function (PDF) of the mutual information (conditioned on the imperfect CSIT), which is highly non-trivial. The conditional PER  usually has no closed-form expression and is non-convex w.r.t. the optimization variables. In \cite{chancesss}, \cite{Chancecon1}, the authors use Bernstein approximation to obtain a conservative convex approximation of the \emph{affine chance constraints}\footnote{Affine chance constraints involve linear forms of random variables \cite{chancesss}.}. These works cannot be used in our problem because the packet error probability involves a quadratic form of random variables. Furthermore, a fixed target PER is assumed  in the existing works \cite{sumgood2}, \cite{Chancecon1}, which is  suboptimal for delay considerations. 	

	\item \textbf{Challenge due to the Average Delay Constraints:} 
	The presence of the average delay constraints fundamentally changes our problem to a stochastic optimization.  Furthermore, due to the imperfect CSIT, the associated optimization problem is  a \emph{partially observed MDP} (POMDP) \cite{pomdp1}, which is  more difficult than regular MDP. A key obstacle of solving the MDP/POMDP is to obtain the relative value function in the associated Bellman equation. Yet, standard solutions can only give numerical solutions to the relative value function \cite{DP_Bertsekas}, which suffer from the issues of slow convergence and lack of insights. It is desirable to obtain a closed-form approximation of the relative value function  in order to have low complexity solutions for our problem. 
\end{itemize}

In this paper, we model the delay-constrained beamforming control  problem as an infinite horizon average cost POMDP. \tcb{By exploiting the special structure in our problem, we derive an equivalent Bellman equation to solve the POMDP.  We then introduce a \emph{virtual continuous time system} (VCTS) and  show that the solution to the associated total cost problem is asymptotically optimal to the POMDP problem when the slot duration is much less than the timescale of the queue evolution.} To deal with the curse of dimensionality induced by the mutual queue coupling, we  leverage  the fact that the CSIT error in practical MU-MIMO systems is usually  small\footnote{A MU-MIMO system with large CSIT errors is not  meaningful as the associated inter-user interference will severely limit the performance of spatial multiplexing.}. As a result, we  adopt perturbation analysis  w.r.t. the CSIT errors and derive a closed-form approximation to the relative value function and analyze the approximation error.  To deal with the challenge of the packet error probability due to the imperfect CSIT, we obtain a tractable closed-form approximation of the conditional PER using a Bernstein-type inequality for quadratic forms of complex Gaussian random variables \cite{bernsteininequ}. Unlike most existing works where the target PER is fixed \cite{Chancecon1}, \cite{fixedper},  the conditional PER of the proposed solution can be dynamically adjusted  according to the observed system state  (CSIT, QSI).   Finally, based on the closed-form approximations of the relative value function and the conditional PER, we derive a low complexity solution using the semidefinite relaxation (SDR) technique and show that the proposed solution achieves significant performance gain over various baseline schemes.

\section{system model}
 In this section, we elaborate on the physical layer model and the bursty source model for the downlink MU-MIMO  system.

\begin{figure}[!htbp]
  \centering
  \includegraphics[width=3.5in]{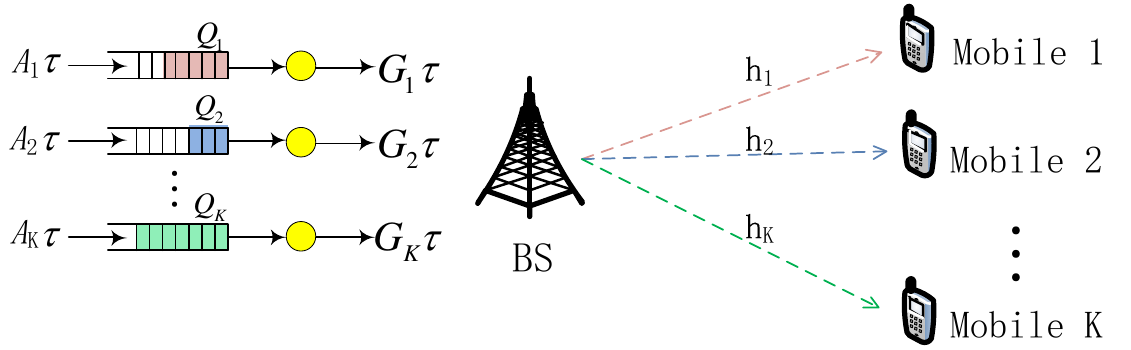}
  \caption{System model of a downlink MU-MIMO  system with $K$ mobiles.}
  \label{mimobc}
\end{figure}

\subsection{MIMO Channel and Imperfect CSIT Models}
We consider a downlink MU-MIMO  system\footnote{When $N_t < K$, there will be a user selection control to select at most $N_t$ active users from the $K$ users and the proposed solution framework could be easily extended to accommodate this user selection control as well.} where a multi-antenna base station (BS) communicates to $K$ single-antenna mobiles  as illustrated in Fig.~\ref{mimobc}. Specifically, the BS is equipped with $N_t \geq K$ antennas.  Let  $\mathbf{h}_k \in \mathbb{C}^{1 \times N_t  }$ be the complex fading coefficient (CSI) from the BS to the $k$-th mobile. Let  $\mathbf{H}=\left\{\mathbf{h}_k:  \forall k \right\}$ denote the global CSI.  In this paper,  the time dimension is partitioned into decision slots indexed by $t$ with  slot duration $\tau$. We have the following assumption on the CSI $\mathbf{H}$.

\begin{Assumption}	[CSI Model]	\label{CSIassum}	
$\mathbf{h}_k\left(t\right)$ remains constant within each decision slot  and is  i.i.d. over  slots for all $k$. Specifically, each element of $\mathbf{h}_k\left(t\right)$ follows a complex Gaussian distribution with zero mean and unit variance\footnote{The assumption on the   fading coefficient could be justified in many applications. For example,  in frequency hopping  systems, the channel fading remains constant within one slot (hop) and is i.i.d. over slots (hops) when the frequency is hopped from one channel to another.}. Furthermore, $\mathbf{h}_k\left(t\right)$ is  independent w.r.t. $k$.~\hfill\IEEEQED
\end{Assumption}


We consider a TDD system. We assume perfect CSI at the mobiles and imperfect CSI at the BS (imperfect CSIT). The imperfect CSIT is due to the channel estimation noise on the uplink pilots or the outdatedness resulting from TDD duplexing delay. Let $\hat{\mathbf{h}}_k \in \mathbb{C}^{1 \times N_t}$ be the imperfect estimate of $\mathbf{h}_k$ at the BS.  Let $\hat{ \mathbf{H}} = \{\hat{\mathbf{h}}_k:  \forall k \}$ denote the global CSIT. We assume MMSE prediction is used at the BS to obtain the CSIT. Therefore,  we  have the following assumption on the imperfect CSIT  $\hat{ \mathbf{H}}$.

\begin{Assumption}	[Imperfect CSIT Model]\label{impCSITm}
The  CSIT $\hat{\mathbf{h}}_k$  is  given by
\begin{align}
	\hat{\mathbf{h}}_k = \mathbf{h}_k + \Delta_k		\label{CSITqua}
\end{align}
where $\Delta_k = \sqrt{\epsilon_k} \mathbf{v}_k$ is the CSIT error and $\mathbf{v}_k$  is a complex Gaussian random vector with zero mean and covariance matrix $\textbf{I}_{N_t}$, i.e., $\mathbf{v}_k \sim \mathcal{CN}\left(0, \textbf{I}_{N_t} \right)$.  $\epsilon_k \geq 0$ is the CSIT error variance, which measures the CSIT quality. Furthermore, $\mathbb{E}\big[\Delta_k  \hat{\mathbf{h}}_k^{\dagger} \big]=0$ by  the orthogonality principle of MMSE \cite{mmse}, where $\left( \cdot \right)^{\dagger}$ denotes the conjugate transpose.~\hfill\IEEEQED
\end{Assumption}
\begin{Remark}	[\tcb{Physical Meaning of CSIT Error Variance $\epsilon_k$}]
	\tcb{Since MMSE estimation is used  to obtain the CSIT based on the uplink pilots, we have $\Delta_k=\frac{\sqrt{E_p}}{1+E_p}z_k^{pilot}-\frac{1}{1+E_p}\mathbf{h}_k$, where $E_p$ is the uplink pilot SNR and $z_k^{pilot}\sim\mathcal{CN}(0, \mathbf{I}_{N_t})$ is the AWGN noise in the received samples of the uplink pilots \cite{mmse}, \cite{noisevar}. Therefore, $\epsilon_k=\frac{1}{1+E_p} \in (0,1]$. In particular, when $\epsilon_k=0$ ($E_p \rightarrow \infty$),  we have $\hat{\mathbf{h}}_k=\mathbf{h}_k$. This corresponds to   the perfect CSIT case. When $\epsilon_k=1$ ($E_p \rightarrow 0$), we have $\mathbb{E}\big[\mathbf{h}_k  \hat{\mathbf{h}}_k^{\dagger} \big]=0$. This corresponds to the no CSIT case.~\hfill\IEEEQED}
\end{Remark}
According to the imperfect CSIT model in Assumption \ref{impCSITm},  the CSIT error kernel is given by the following conditional PDF:
\begin{align}	\label{csiterrknl}
	\Pr \big[ \hat{\mathbf{h}}_k | \mathbf{h}_k \big] = \frac{1}{\pi \epsilon_k} \exp \left(-\frac{| \hat{\mathbf{h}}_k -\mathbf{h}_k |^2}{\epsilon_k}\right)
\end{align}

Let $s_k$  denote the information symbol for the $k$-th mobile. The transmitted signal is given by $\sum_{k=1}^K \mathbf{w}_k s_k$ where $\mathbf{w}_k\in \mathbb{C}^{ N_t \times 1}$ is the transmit beamforming vector for $s_k$. Therefore, the  received signal   at the $k$-th mobile  is given by\footnote{Note that $\|\mathbf{w}_k\|^2$ is the power allocated for information symbol $s_k$.}
\begin{align}	\label{recvsig}
	y_k = \underbrace{\mathbf{h}_k \mathbf{w}_k s_k }_{\text{desired signal}}+ \underbrace{\sum_{j \neq k} \mathbf{h}_k \mathbf{w}_j  s_j}_{\text{interference signal }} + \underbrace{z_k}_{\text{noise}}
\end{align}
where $z_k \sim \mathcal{CN}\left(0,1\right)$ is the  i.i.d. complex Gaussian channel noise.

\subsection{Mutual Information and System Goodput}
For  given CSI realization $ \mathbf{H}$ and collection of the beamforming control actions  of all the $K$ flows $\mathbf{w} \triangleq \big\{\mathbf{w}_k: \forall k \big\}$, the   mutual information (bit/s/Hz) between the BS and the $k$-th mobile  is given by
\begin{align}
	C_k\left(\mathbf{H},\mathbf{w} \right)= \log\left(1+ \frac{\big|\mathbf{h}_k  \mathbf{w}_k\big|^2 }{1+\sum_{j \neq k} \big| \mathbf{h}_k \mathbf{w}_j\big|^2}\right)		\label{mutualinfo}
\end{align}

Let $R_k$ be the transmit data rate for the $k$-th mobile at the BS. Due to imperfect CSIT, there is uncertainty of the mutual information $C_k\left(\mathbf{H}, \mathbf{w} \right)$ due to the imperfect CSIT $\hat{\mathbf{H}}$. As a result, the  goodput \cite{sumgood2}, i.e.,  the bit/s/Hz successfully delivered to the mobile when the transmit data rate is $R_k$ is given by
\begin{align}	\label{goodput}
	& G_k\left(\mathbf{H},\mathbf{w} \right) = R_k\mathbf{1}\left(R_k\leq C_k\left(\mathbf{H},\mathbf{w} \right)\right)
\end{align}
where $\mathbf{1}(\cdot)$ is the indicator function.

\subsection{Bursty Source Model and Queue Dynamics}	\label{queuedym}
As illustrated in Fig.~\ref{mimobc}, the BS maintains $K$ data queues for the bursty traffic flows towards the $K$ mobiles. Let \tcb{$\mathbf{A}\left( t\right)=\left(A_1\left( t\right)\tau,\dots,  A_K\left( t\right)\tau \right)$} be the {random data arrivals (number of bits)}  at the end of the $t$-th decision slot  for the $K$ mobiles. We have the following assumption on  {$\mathbf{A}\left( t\right)$}.
\begin{Assumption} [Bursty Source Model]	\label{assumeA}
	Assume that $A_k\left(t\right)$ is i.i.d. over decision slots according to a general distribution $\Pr[A_k]$. The moment generating function of $A_k$ exists with  mean $\mathbb{E}[A_k]=\lambda_k$.  $ A_k\left(t\right)$ is independent w.r.t. $k$. \tcb{Furthermore, each arrival packet for the $k$-th queue contains $R_k\tau$ bits}\footnote{\tcb{In practical systems such as UMTS or LTE, there is a segmentation process in the MAC layer such that the data packets from the MAC layer  match the payload size of the PHY layer packets.}} and  $(\lambda_1,\dots,\lambda_K)$ lies within the stability region \cite{capacityregion} of the  system.~\hfill\IEEEQED
\end{Assumption}

Let $Q_k\left( t\right) \in \mathcal{Q}$  denote the QSI (number of bits) at the $k$-th queue of the BS at the beginning of the $t$-th slot, where {$\mathcal{Q}=[0, \infty)$} is the   QSI state space. Let $\mathbf{Q}\left( t\right) = \left(Q_1\left(t \right),\dots,Q_K \left( t\right)  \right) \in \boldsymbol{\mathcal{Q}} \triangleq \mathcal{Q}^K$ denote the global QSI.  Furthermore,   we assume there is ACK/NAK feedback\footnote{The $R_k \tau$ information bits will be removed from the $k$-th queue at the BS only when the bits are successfully received by the MS (via ACK feedback). Otherwise, the information bits will be maintained at the  queue and wait for subsequent transmission opportunity.}  from the mobiles to the BS. Hence, the dynamics of the $k$-th queue at the BS  is given by\footnote{We assume that the controller at the BS is  causal so that new arrivals are observed after the  control actions are performed at each decision slot.}
\begin{align}	\label{queue_sys}
	Q_k\left( t+1\right) = \left[Q_k\left( t\right)- G_k\left(\mathbf{H}\left(t\right),\mathbf{w}\left(t\right) \right) \tau\right]^+   + \tcb{A_k\left(t \right)\tau}
\end{align}
where $[x]^+=\max\{0,x\}$.
\begin{Remark} [Coupling Property of  Queue Dynamics] 	\label{coup_rem}
	 The $K$ queue dynamics in the downlink MU-MIMO system are coupled together due to the  interference term in (\ref{recvsig}). Specifically,  the departure of the $k$-th queue depends on the beamforming control actions of all the other data flows. ~\hfill\IEEEQED
\end{Remark}

\section{Delay-constrained  Control Problem Formulation}
In this section, we  formally define  the  beamforming   control policy  and formulate the delay-constrained control problem for the downlink MU-MIMO  system.
\subsection{Beamforming   Control Policy}
At the beginning of  each decision slot, the BS determines the  beamforming  control actions based on the global observed system state $\big( \hat{\mathbf{H}}, \mathbf{Q}\big)$ according to the following  stationary control policy.
\begin{Definition} [Stationary  Beamforming  Control Policy]	\label{deff1}
	A stationary  beamforming   control policy for the $k$-th data flow $\Omega_k$ is a mapping from the global observed system state $\big( \hat{\mathbf{H}}, \mathbf{Q}\big)$ to the  beamforming   control actions of the $k$-th data flow. Specifically, we have  $\Omega_k\big( \hat{\mathbf{H}}, \mathbf{Q}\big)= \mathbf{w}_k \in \mathbb{C}^{N_t \times 1}$. Furthermore, let $\Omega= \{ \Omega_k:\forall  k \}$ denote the aggregation of the control policies for all the $K$ data flows.~\hfill\IEEEQED
\end{Definition}

For notation convenience, we  denote $\boldsymbol{\chi}=\big(\mathbf{H}, \hat{\mathbf{H}}, \mathbf{Q}\big)$ as the global system state.  Given a  control policy $\Omega$, the induced random process $\left\{\boldsymbol{\chi}\left(t \right)\right\}$ is a controlled Markov chain with the following transition probability\footnote{\tcb{The first equality of  (\ref{trankernel})  is due to the i.i.d. assumption of the CSI model and the assumption of the imperfect CSIT model. The second equality  is due to the independence between $\mathbf{Q}\left(t+1\right)$ and $\hat{\mathbf{H}}\left(t\right)$ conditioned on $\mathbf{H}(t)$, $\mathbf{Q}(t)$ and  $\Omega\big(\hat{\mathbf{H}}(t), \mathbf{Q}(t)  \big)$.}}:
\begin{align}	\label{trankernel}
	& \Pr\big[ \boldsymbol{\chi}\left(t+1 \right) \big| \boldsymbol{\chi}\left(t \right), \Omega\big( \hat{\mathbf{H}}(t), \mathbf{Q}(t)\big)\big] \notag	\\
	\tcb{=}& \tcb{\Pr\big[ \hat{\mathbf{H}}\left(t+1  \right), \mathbf{H}\left(t+1\right)\big]}	\notag \\
	 &\tcb{ \cdot \Pr\big[ \mathbf{Q}\left(t+1\right) \big| \boldsymbol{\chi}\left(t \right), \Omega\big( \hat{\mathbf{H}}(t), \mathbf{Q}(t) \big)\big]}	\notag \\
	=& \Pr \big[\mathbf{H}\left(t+1 \right) \big]   \Pr \big[\hat{\mathbf{H}}\left(t+1  \right) \big| \mathbf{H}\left(t+1\right) \big]\notag \\
	& \cdot \Pr \big[ \mathbf{Q}\left(t+1\right) \big|  \mathbf{H}(t), \mathbf{Q}(t), \Omega\big(\hat{\mathbf{H}}(t), \mathbf{Q}(t)  \big) \big]
\end{align}
where the queue transition probability is given by
\begin{align}	
&\Pr \big[ \mathbf{Q}\left(t+1\right) \big|  \mathbf{H}(t), \mathbf{Q}(t), \Omega\big(\hat{\mathbf{H}}(t), \mathbf{Q}(t)  \big) \big] \notag \\
=&
 \left\{
	\begin{aligned}	 
		&\prod_{k} \Pr\big[A_k\left( t\right)\big]   \ \ \tcb{\text{if } Q_k\left( t+1\right) \text{is given by (\ref{queue_sys})}, \ \forall  k} \\
		&\ 0 \hspace{2.4cm} \text{otherwise}
	   \end{aligned}
   \right.
  \end{align}
Note that the transition kernel in (\ref{trankernel}) is time-homogeneous due to the i.i.d. property of the arrival $A_k\left(t\right)$   in Assumption \ref{assumeA}.   Furthermore, {we have the following definition on the admissible  control policy}. 
\begin{Definition}	[{Admissible Control Policy}]	\label{admissibledis}
	{A policy $\Omega$ is  admissible if the following requirements are satisfied:}
	\begin{itemize}
		\item $\Omega$ is a unichain policy, i.e., the controlled Markov chain $\left\{\boldsymbol{\chi}\left(t \right)\right\}$ under $\Omega$ has a single recurrent class (and possibly some transient states) \cite{DP_Bertsekas}.
		\item {The queueing system under $\Omega$ is  stable in the sense that $\lim_{t \rightarrow \infty} \mathbb{E}^{\Omega} \big[\sum_{k=1}^K Q_k^2(t) \log Q_k(t) \big] < \infty$, where $\mathbb{E}^{\Omega} $ means taking expectation w.r.t. the probability measure induced by the control policy $\Omega$.}~\hfill\IEEEQED
	\end{itemize}
\end{Definition}

\subsection{Problem Formulation}
As a result, under {an admissible}  control policy $\Omega$, the average delay cost of the $k$-th data flow starting from a given initial  state $\boldsymbol{\chi}\left(0\right)$ is given by
\begin{align}	\label{delay_cost}
	\overline{D}_k^\Omega\left(\boldsymbol{\chi}\left(0 \right)\right)  = \limsup_{T \rightarrow \infty} \frac{1}{T} \sum_{t=0}^{T-1} \mathbb{E}^{\Omega} \left[\frac{Q_k\left(t\right)}{\lambda_k} \right], \quad \forall k	
\end{align}	
Similarly, under {an admissible} control policy $\Omega$,  the average power cost of the BS  starting from a given initial  state $\boldsymbol{\chi}\left(0\right)$ is given by
\begin{align}		\label{power_cost}
	\overline{P}^{\Omega}\left(\boldsymbol{\chi}\left(0 \right)  \right) = \limsup_{T \rightarrow \infty} \frac{1}{T} \sum_{t=0}^{T-1} \mathbb{E}^{\Omega}  \bigg[ \sum_{k=1}^K \left\|\mathbf{w}_k \left(t \right)\right\|^2  \bigg]
\end{align}
In general, we are interested in  minimizing either the average delay  or  the average power and both cannot be minimized at the same time. As a result, we consider the following formulation which can achieve a Pareto optimal tradeoff between the average delay costs and average power cost. 
\begin{Problem}	[Delay-Constrained Beamforming Control]  \label{IHAC_MDP}
	For any initial  system state   $\boldsymbol{\chi}(0)$, the delay-constrained beamforming  control problem is formulated as 
\begin{eqnarray} \
	\underset{\Omega}{\min} &&L_{\boldsymbol{\gamma}}^{\Omega}\left( \boldsymbol{\chi}\left(0 \right)\right) 	\notag \\
	&=&\sum_{k=1}^K \gamma_k \overline{D}_k^\Omega\left(\left(\boldsymbol{\chi}\left(0 \right)\right)\right) +  \overline{P}^{\Omega}\left(\boldsymbol{\chi}\left(0 \right)\right) \notag \\
	&=& \limsup_{T \rightarrow \infty} \frac{1}{T} \sum_{t=0}^{T-1} \mathbb{E}^{\Omega}  \left[c\big(\mathbf{Q}\left(t\right), \Omega\big(\hat{\mathbf{H}}(t), \mathbf{Q}(t)  \big)\big)    \right]	\notag	
\end{eqnarray}
where $c\left(\mathbf{Q}, \mathbf{w}\right)=\sum_{k=1}^K \big(\|\mathbf{w}_k \|^2+\gamma_k \frac{Q_k}{\lambda_k} \big)$ and  $\boldsymbol{\gamma}=\left\{\gamma_k >0: \forall k \right\}$ is the delay price\footnote{\tcb{The delay price $\boldsymbol{\gamma}$ indicates the relative importance of the delay requirement over the average power. Larger values of $\boldsymbol{\gamma}$ correspond to greater importance in delay.} $\boldsymbol{\gamma}$ can also be interpreted as the corresponding Lagrange Multipliers associated with the  delay constraints of the $K$ information flows \tcb{\cite{equiformu}}.} of the $K$ information flows.~\hfill\IEEEQED
\end{Problem}


For a given positive delay price $\boldsymbol{\gamma}$, the solution to Problem \ref{IHAC_MDP} corresponds to a point on the Pareto optimal tradeoff curve between  the average delay costs $\overline{D}_1^{\Omega}, \cdots, \overline{D}_K^{\Omega}$ and the average power cost $\overline{P}^{\Omega}$.  Problem \ref{IHAC_MDP} is an infinite horizon average cost  POMDP \cite{surveydelay},  because the controller (i.e., the BS) only has partial observation of the system state (imperfect CSIT and QSI). Note that POMDP is well-known to  be a very difficult problem \cite{surveydelay}. In the next  subsection, by exploiting the special structure in our problem, we derive an equivalent Bellman equation to simplify the   POMDP problem.

\subsection{General Solution to the Optimal Control Problem}
We first define the partitioned actions below.
\begin{Definition}	[Partitioned Actions]	\label{paract}
	Given a control policy $\Omega$, we define $\Omega\left(\mathbf{Q} \right)=\big\{ \Omega_k\left(\mathbf{Q} \right): \forall k\big\}$, where $\Omega_k\left(\mathbf{Q} \right)=\big\{\mathbf{w}_k= \Omega_k\big( \hat{\mathbf{H}}, \mathbf{Q}\big): \forall \hat{\mathbf{H}}\big\}$ is the  collection of actions of the $k$-th flow for all possible CSIT $\hat{\mathbf{H}}$ conditioned on a given QSI $\mathbf{Q}$. The complete control policy  is therefore equal to the union of all partitioned actions, i.e., $\Omega=\bigcup_{\mathbf{Q}}\Omega\left(\mathbf{Q} \right)$.~\hfill\IEEEQED
\end{Definition}

While the POMDP in Problem \ref{IHAC_MDP} is  difficult in general, we utilize Definition \ref{paract} and the i.i.d. assumption of the CSI to derive an \emph{equivalent Bellman equation} as summarized below. 
\begin{Theorem} [\tcb{Sufficient Conditions for Optimality}]	\label{LemBel}
	For any given  $\boldsymbol{\gamma}$, assume there exists a \tcb{($\theta^\ast, \{ V^\ast\left(\mathbf{Q} \right) \}$)} that solves the following \emph{equivalent Bellman equation}:
	{\begin{align}	
		 &\theta^\ast {\tau} + V^\ast \left(\mathbf{Q} \right) \hspace{3cm} \forall \mathbf{Q} \in \boldsymbol{\mathcal{Q}} \label{OrgBel} \\
		 = &\min_{\Omega\left(\mathbf{Q} \right)} \bigg[ \widetilde{c}\left(\mathbf{Q}, \Omega\left(\mathbf{Q} \right)\right) {\tau} +  \sum_{\mathbf{Q}'}\Pr\left[ \mathbf{Q}'| \mathbf{Q},  \Omega\left(\mathbf{Q} \right)\right]V^\ast \left(\mathbf{Q} '\right)    \bigg]\notag
	\end{align}}where $\widetilde{c}\left( \mathbf{Q}, \Omega\left(\mathbf{Q} \right)\right) =\mathbb{E}\big[ c\big(\mathbf{Q}, \Omega\big(\hat{\mathbf{H}},\mathbf{Q}\big)\big)\big| \mathbf{Q}  \big]$ is the per-stage cost function, and $\Pr\left[ \mathbf{Q}'| \mathbf{Q},  \Omega\left(\mathbf{Q} \right)\right] = \\ \mathbb{E}\big[\Pr \big[ \mathbf{Q}'\big| \mathbf{H}, \mathbf{Q}, \Omega\big(\hat{\mathbf{H}}, \mathbf{Q}  \big) \big]  \big| \mathbf{Q} \big]$ is the transition kernel. \tcb{Futhermore,  for all  admissible control policy $\Omega$ and initial queue state $\mathbf{Q}\left(0 \right)$, $V^\ast$ satisfies the following \emph{transversality condition}:
	\begin{align}	\label{transodts}
	\lim_{T \rightarrow \infty} \frac{1}{T}\mathbb{E}^{\Omega}\left[ V^\ast\left(\mathbf{Q}\left(T \right) \right) |\mathbf{Q}\left(0 \right)\right]=0
\end{align}}Then, ${\theta^\ast}=\underset{\Omega}{\min} L_{\boldsymbol{\gamma}}^{\Omega}\left( \boldsymbol{\chi}\left(0 \right)\right) $ is the optimal average cost for any initial state  $\boldsymbol{\chi}\left(0 \right) $ and $V^\ast\left(\mathbf{Q}\right)$ is called  \emph{relative value function}. If $\Omega^{\ast}\left( \mathbf{Q} \right)$ attains the minimum of the R.H.S. of (\ref{OrgBel}) for all $\mathbf{Q} \in \boldsymbol{\mathcal{Q}} $, then $\Omega^{\ast}$ is the optimal control policy  for Problem \ref{IHAC_MDP}.~\hfill\IEEEQED	
\end{Theorem}

\begin{proof}
	please refer to Appendix A.
\end{proof}
\begin{Remark}	[Interpretation of {Theorem} \ref{LemBel}]
	The equivalent Bellman equation in (\ref{OrgBel}) is defined on the  queue state $\mathbf{Q}$ only. Nevertheless, the optimal beamforming   control policy $\Omega^{\ast}$ obtained by solving (\ref{OrgBel}) is still adaptive to the  observed state $\big( \hat{\mathbf{H}}, \mathbf{Q}\big)$. {Furthermore, based on the unichain assumption  of the control  policy, the solution obtained from Theorem 1 is unique \cite{DP_Bertsekas}.} ~\hfill\IEEEQED
\end{Remark}

Based on Theorem \ref{LemBel}, we establish the following corollary on  the approximation of  the Bellman equation in (\ref{OrgBel}).
\tcb{\begin{Corollary}	[Approximate Bellman Equation]	\label{cor1}
	For any given  $\boldsymbol{\gamma}$, if 
	\begin{itemize}
		\item	there is a unique  ($\theta^\ast, \{ V^\ast\left(\mathbf{Q} \right) \}$) that satisfies the Bellman equation and transversality condition in Theorem \ref{LemBel}.
		\item	there exists  $\theta$ and $V\left( \mathbf{Q}\right)$ of class\footnote{$f(\mathbf{x})$ ($\mathbf{x}$ is a $K$-dimensional vector) is of class $\mathcal{C}^2(\mathbb{R}_+^K)$, if the  first and second order partial derivatives of $f(\mathbf{x})$ w.r.t. each element of $\mathbf{x}$ are  continuous when $\mathbf{x}\in \mathbb{R}_+^K$.} $\mathcal{C}^2(\mathbb{R}_+^K)$ that solve the following \emph{approximate Bellman equation}:
	\begin{align}		\label{conperr}
		&\theta = \min_{ \Omega\left( \mathbf{Q} \right)} \bigg[ \widetilde{c}\left(\mathbf{Q}, \Omega\left(\mathbf{Q} \right)\right) +  \sum_{k=1}^K \frac{\partial V \left(\mathbf{Q} \right) }{\partial Q_k} \Big[  \lambda_k  - \mathbb{E}\left[R_k \right.  \notag \\
		&\left. \left(  1-\Pr\big[R_k> C_k\big(\mathbf{H},\Omega(\hat{\mathbf{H}},\mathbf{Q}) \big) \big| \hat{\mathbf{H}},\mathbf{Q}\big]\right) \big| \mathbf{Q}\right]   \Big] \bigg]
	\end{align}where  $\Pr\big[R_k> C_k\left(\mathbf{H}, \mathbf{w} \right) \big| \hat{\mathbf{H}},\mathbf{Q} \big]$ is the conditional PER (conditioned on the observed state $\big(\hat{\mathbf{H}},\mathbf{Q}\big)$). Furthermore,  for all admissible  control policy $\Omega$ and initial queue state $\mathbf{Q}\left(0 \right)$, the transversality condition  in (\ref{transodts}) is satisfied for $V$. 
	\end{itemize}
Then, we have
	\begin{align}	\label{simbelman}
		\theta^\ast&=\theta+o(1)	\\
		V^\ast\left(\mathbf{Q} \right)&=V\left(\mathbf{Q} \right)+o(1), \quad \forall \mathbf{Q} \in \boldsymbol{\mathcal{Q}}	\label{15resu}
	\end{align}where the error term $o(1)$ asymptotically goes to zero  for sufficiently small slot duration $\tau$.~\hfill\IEEEQED
\end{Corollary}}
\begin{proof}
	please refer to Appendix B.
\end{proof}

\tcb{Corollary \ref{cor1} states that the difference between ($\theta, \{ {V}\left(\mathbf{Q} \right) \}$) obtained by solving (\ref{conperr}) and ($\theta^\ast, \{ {V^\ast}\left(\mathbf{Q} \right) \}$) in (\ref{OrgBel}) is asymptotically small w.r.t. the slot duration $\tau$. Therefore, we can focus on solving the approximate Bellman equation in  (\ref{conperr}), which is a simpler problem than solving the original Bellman equation in (\ref{OrgBel}).}

There are two technical obstacles to {solving the approximate Bellman equation in (\ref{conperr})}. Firstly, deriving the optimal control policy from (\ref{conperr}) requires  knowledge of the relative value function $V\left(\mathbf{Q} \right)$. \tcb{In fact, the relative value function captures the urgency information of each data flow and plays a key role in  delay-aware  control.} However, obtaining the relative value function is not trivial as it involves solving  a large  system of nonlinear fixed point equations.  Brute-force approaches to solve these fixed point equations such as  value iteration and policy iteration \cite{DP_Bertsekas} have huge complexity.  Secondly, deriving the optimal control policy from (\ref{conperr})  requires  knowledge of the conditional PER. The conditional PER does not have closed-form expression and is not convex  for general beamforming design \cite{fixedper}. To address the first issue, we  introduce the virtual continuous time system (VCTS) and obtain a closed-form approximation of the relative value function in Section \ref{Example22}. To address the second issue, we  apply Bernstein approximation \cite{Chancecon1}, \cite{fixedper} to obtain a tractable convex approximation of the conditional PER   in Section \ref{KnTbased}.

\section{Closed-Form Approximation of Relative Value Function}	\label{Example22}
In this section, we  adopt a continuous time approach to obtain a closed-form approximation of the relative value function $V\left(\mathbf{Q}\right)$. We first reverse-engineer a virtual continuous time system (VCTS) from the original discrete time POMDP (DT-POMDP). We then utilize perturbation theory to obtain the closed-form approximation of $V\left(\mathbf{Q}\right)$.

\subsection{Virtual Continuous Time System} \label{VCTS}
We first define the VCTS, which can be viewed as a characterization of the mean behavior of the DT-POMDP in (\ref{queue_sys}).  The motivation of studying the problem in the continuous domain is to utilize the well-established theories of calculus and differential equations to obtain a closed-form approximation of $V\left(\mathbf{Q}\right)$. The VCTS is a fictitious system with a continuous virtual queue state $\mathbf{q}\left(t\right) = \left( q_1\left(t\right),\dots,q_K\left(t\right)\right) \in \overline{\boldsymbol{\mathcal{Q}}} {\triangleq \overline{\mathcal{Q}}^K}$,  where  $q_k\left(t\right) \in {\overline{\mathcal{Q}}}$ is the   state of the $k$-th virtual queue at time $t$. {$\overline{\mathcal{Q}}\triangleq [0, +\infty)$ denotes the virtual queue state space}. Given an initial   system state $\mathbf{q}(0)\in \overline{\boldsymbol{\mathcal{Q}}} $,  the  trajectory of the $k$-th virtual queue is  described by the following  differential equation:
\begin{align}	\label{VCTSeg2}
	\frac{ \mathrm{d}} {\mathrm{d}t}  q_{k}\left(t\right) = - \mathbb{E}\left[G_k\left(\mathbf{H},\mathbf{w}\left(t\right) \right)  \big| \mathbf{q}\left(t \right)\right]   + \lambda_k
\end{align}
where $G_k\left(\mathbf{H},\mathbf{w} \right)$ is the goodput in (\ref{goodput}) and $ \lambda_k$ is the average data arrival rate in Assumption \ref{assumeA}. 

Let $\Omega_k^v$ be the virtual control policy for the $k$-th  flow  of the VCTS which is a mapping from the global   state  to the   actions of the $k$-th  flow. Specifically, we have  $\Omega_k^v\big(\hat{\mathbf{H}}, \mathbf{q} \big)= \mathbf{w}_k$.  Furthermore, let $\Omega^v= \{ \Omega_k^v: \forall k \}$ be the aggregation of the virtual control policies  for all the $K$ flows in the VCTS.  Similarly, we define the associated partitioned actions  as  $\Omega^v \left(\mathbf{q} \right)=\big\{ \Omega_k^v\left(\mathbf{q} \right): \forall k\big\}$, where $\Omega_k^v\left(\mathbf{q} \right)=\big\{\mathbf{w}_k = \Omega_k^v\big( \hat{\mathbf{H}}, \mathbf{q}\big): \forall \hat{\mathbf{H}}\big\}$.

From the VCTS dynamics in (\ref{VCTSeg2}), there exists a steady state $\mathbf{q}^{\infty}=(0, \dots, 0)$  for   $\mathbf{q}(t)$, i.e., $ \lim_{t \rightarrow \infty} \mathbf{q}(t) \\ =\mathbf{q}^{\infty}$. The virtual control action that maintains $\mathbf{q}(t)$ at the steady state $\mathbf{q}^\infty$ is defined as the \emph{steady state control action}.
\begin{Definition}	[Steady State Control Action]	\label{defsteadyconrt}
	A control action $\mathbf{w}^{\infty}\triangleq \big\{\mathbf{w}_k^{\infty}: \forall k, \hat{\mathbf{H}}\big\}$ is called a \emph{steady state control action} if it satisfies $\mathbb{E}\left[G_k\left(\mathbf{H},\mathbf{w}^{\infty} \right)  \big| \mathbf{q}^{\infty}\right]    =\lambda_k$ $(\forall k)$.~\hfill\IEEEQED
\end{Definition}

For the VCTS, we consider the following \emph{virtual} per-stage cost function:
\begin{align}
	\overline{c}\left(\mathbf{q}, \Omega^v(\mathbf{q})\right) =\mathbb{E}\left[\sum_{k=1}^K\big( \|\mathbf{w}_k\|^2+\gamma_k \frac{|q_k|}{\lambda_k} \big)\Bigg|\mathbf{q} \right]-c^{\infty}
\end{align}
where $c^{\infty}=\mathbb{E}\big[\sum_{k=1}^K \|\mathbf{w}_k^{\infty}\|^2\big| \mathbf{q}^{\infty}\big]$ and $\mathbf{w}^{\infty}$ is the steady state control action  as defined in Definition \ref{defsteadyconrt}. Furthermore, we have the following definition on the admissible  virtual control policy for the VCTS. 
\begin{Definition}	[Admissible Virtual Control Policy for VCTS]		\label{addvctss}
	A virtual policy $\Omega^v$ for the VCTS is  admissible if  the following requirements are satisfied:
	\begin{itemize}
		\item  For any  initial state $\mathbf{q}(0)\in \overline{\boldsymbol{\mathcal{Q}}}$, the virtual queue trajectory $\mathbf{q}(t)$ in (\ref{VCTSeg2}) under $\Omega^v$ is unique.
		\item For any  initial state $\mathbf{q}(0)\in \overline{\boldsymbol{\mathcal{Q}}}$, the total cost $\int_0^{\infty} \overline{c}\left(\mathbf{q}\left(t\right), \Omega^v\left(\mathbf{q}\left(t\right)\right)\right) \  \mathrm{d}t$ under $\Omega^v$ is bounded.~\hfill\IEEEQED
	\end{itemize}
\end{Definition}

Given an admissible  control policy $\Omega^v$, we define the total cost of the VCTS starting from a given initial global virtual queue state $\mathbf{q}\left(0 \right)$ as
\begin{equation}		\label{totalU}
	\tcb{J^{\Omega^v} \left(\mathbf{q}\left(0\right);\boldsymbol{\epsilon} \right) } = \int_0^{\infty} {\overline{c}\left(\mathbf{q}\left(t\right), \Omega^v\left(\mathbf{q}\left(t\right)\right)\right)} \  \mathrm{d}t, \quad \mathbf{q}\left(0\right)  \in \overline{\boldsymbol{\mathcal{Q}}}
\end{equation}
where  $\boldsymbol{\epsilon}  \triangleq \{\epsilon_k: \forall k\}$ \tcb{is a coupling parameter which affects the virtual queue evolution in (\ref{VCTSeg2})}. We consider an infinite horizon total cost problem associated with the VCTS as below.

\begin{Problem}		[Infinite Horizon Total Cost Problem for  VCTS]  \label{fluid problem1}
For any initial  virtual queue state  $\mathbf{q}(0)  \in \overline{\boldsymbol{\mathcal{Q}}}$, the infinite horizon total cost problem for the VCTS is formulated as
	\begin{align}
		\min_{\Omega^v} \tcb{J^{\Omega^v} \left(\mathbf{q}\left(0\right);\boldsymbol{\epsilon} \right)}
	\end{align}
	where \tcb{$J^{\Omega^v} \left(\mathbf{q}\left(0\right);\boldsymbol{\epsilon} \right)$} is given in (\ref{totalU}).~\hfill\IEEEQED
\end{Problem}

Note that the  two technical conditions in Definition \ref{addvctss}  on  the admissible policy  are for the existence of an optimal policy for the total cost problem in Problem \ref{fluid problem1}.  The above total cost problem has been well-studied in the continuous time optimal control theory \cite{DP_Bertsekas}. The solution can be obtained by solving the \emph{Hamilton-Jacobi-Bellman} (HJB) equation as  below. 
\begin{Lemma}	[\tcb{Sufficient Conditions for Optimality under VCTS}]	\label{HJB1}
	\tcb{Assume there exists a function $J\left( \mathbf{q};\boldsymbol{\epsilon}\right)$ of class $\mathcal{C}^2(\mathbb{R}_+^K)$ that solves the following HJB equation:}
	\begin{align}	\label{cenHJB}
		&\min_{\Omega^v \left(\mathbf{q} \right)} \mathbb{E}\left[ \sum_{k=1}^K \left( \|\mathbf{w}_k\|^2+\gamma_k \frac{|q_k|}{\lambda_k} \right)- c^{\infty}  \right.  \notag \\
		& \left.   + \sum_{k=1}^K \left( \frac{\partial  J \left(\mathbf{q};\boldsymbol{\epsilon}\right)} {\partial q_k} \left(- G_k\left(\mathbf{H},\mathbf{w}\right)  + \lambda_k\right) \right)\Bigg| \mathbf{q} \right] =0
	\end{align}
with  boundary condition \tcb{$J\left(\mathbf{0};\boldsymbol{\epsilon}\right)=0$.  Furthermore, $J$ satisfies the following  conditions: 
\begin{enumerate} [1)]
		\item  $\lim_{t \rightarrow \infty} J\left(\mathbf{q}\left(t\right);\boldsymbol{\epsilon}\right) \leq 0$ for all admissible control policy $\Omega^v$ and initial condition $\mathbf{q}\left(0\right)=\mathbf{q} \in \overline{\boldsymbol{\mathcal{Q}}}$.
		\item $\lim_{t \rightarrow \infty} J\left(\mathbf{q}^{\ast}\left(t\right);\boldsymbol{\epsilon}\right) =0$ for a given control $ \Omega^{v\ast}$ and the corresponding state trajectory $\mathbf{q}^{\ast}\left(t\right)$, where $\Omega^{v \ast} \left(\mathbf{q} \right)$ achieves the minimum of the L.H.S. of (\ref{cenHJB}) for any $\mathbf{q} \in \overline{\boldsymbol{\mathcal{Q}}}$.
\end{enumerate}}Then, \tcb{$J\left(\mathbf{q};\boldsymbol{\epsilon}\right) =\min_{\Omega^v} J^{\Omega^v} \left( \mathbf{q} \left(0\right);\boldsymbol{\epsilon}\right)$} is the optimal total cost when $\mathbf{q}(0)=\mathbf{q}$ and \tcb{$J\left(\mathbf{q};\boldsymbol{\epsilon}\right)$} is called \emph{fluid value function}. $\Omega^{v \ast}$ is the optimal virtual control policy for Problem \ref{fluid problem1}. .~\hfill\IEEEQED
\end{Lemma}
\begin{proof}
	please refer to \cite{DP_Bertsekas} for details.
\end{proof}

In the following lemma, we establish the relationship between the solution of Problem \ref{fluid problem1} ($c^\infty, \{J\left(\mathbf{Q};\boldsymbol{\epsilon}\right)\}$)  and the solution of the Bellman equation ($\theta^\ast, \{V^\ast\left(\mathbf{Q}\right)\}$).

\begin{Theorem}	[\tcb{Relationship  between  ($c^\infty, \{J\left(\mathbf{Q};\boldsymbol{\epsilon}\right)\}$)   and  ($\theta^\ast, \{V^\ast\left(\mathbf{Q}\right)\}$)}]	\label{them111}
	\tcb{Suppose $\left\{\frac{\partial J\left(\mathbf{Q}; \boldsymbol{\epsilon} \right)}{\partial Q_k}: \forall k \right\}$  are increasing functions of all $Q_k$ and $J\left(\mathbf{Q}; \boldsymbol{\epsilon} \right)=\mathcal{O}\left(\sum_{k=1}^K Q_k^2 \log Q_k \right)$. If  ($c^{\infty},  \{ J\left(\mathbf{Q}; \boldsymbol{\epsilon} \right) \}$) satisfies the sufficient conditions in Lemma \ref{HJB1}, then we have $V^\ast\left(\mathbf{Q} \right)=J\left(\mathbf{Q};\boldsymbol{\epsilon}\right)+o(1)$ and $\theta^\ast = c^\infty + o(1)$, where $o(1)$ denotes the asymptotically small error term w.r.t. the slot duration $\tau$.~\hfill\IEEEQED}
\end{Theorem}
\begin{proof}  
please  refer to Appendix	 C.	
\end{proof}

\tcb{The difference between the fluid value function   $J\left(\mathbf{Q};\boldsymbol{\epsilon}\right)$ and the optimal relative value function  $V^\ast\left(\mathbf{Q} \right)$  is  $o(1)$ w.r.t. to the slot duration  $\tau$. Therefore, we can focus on solving the HJB equation in (\ref{cenHJB}) by leveraging the well-established theories of calculus and differential equations.}

\subsection{Perturbation Analysis of \tcb{$ J\left(\mathbf{q};\boldsymbol{\epsilon}\right)$} }	
Deriving \tcb{$ J\left(\mathbf{q};\boldsymbol{\epsilon}\right)$} involves solving  a $K$-dimensional non-linear PDE in (\ref{cenHJB}), which is in general challenging.   To  obtain a closed-form approximation of \tcb{$ J\left(\mathbf{q};\boldsymbol{\epsilon}\right)$}, we treat the VCTS in (\ref{VCTSeg2}) as a \emph{perturbation}  of a \emph{base VCTS}, which is defined below. 
\begin{Definition} [Base VCTS]	\label{base_sys}
	A  base VCTS is the VCTS in (\ref{VCTSeg2}) with $\boldsymbol{\epsilon}=\mathbf{0}$.~\hfill\IEEEQED
\end{Definition}

We first study  the  base VCTS and use \tcb{$J\left(\mathbf{q};\mathbf{0}\right) $} to obtain a closed-form approximation of \tcb{$J\left(\mathbf{q};\boldsymbol{\epsilon}\right)$}.  

For sufficiently large delay price $\boldsymbol{\gamma}$, the optimal beamforming control  for  Problem \ref{fluid problem1} becomes zero-forcing (ZF) beamforming\footnote{This is because  large delay price  corresponds to the high SNR regime. Hence, at the high SNR regime, ZF beamforming is asymptotically optimal \cite{mmse} because there is no interference between the $K$ flows when $\boldsymbol{\epsilon}=0$. }. As a result, the $K$ queue dynamics of the base VCTS are totally decoupled due to the absence of interference when $\boldsymbol{\epsilon}=\mathbf{0}$. In other words, the downlink MU-MIMO system is equivalent to a decoupled system with $K$ independent data flows. We have the following lemma summarizing the fluid value function \tcb{$J\left(\mathbf{q};\mathbf{0}\right)$} of the base VCTS.
\begin{Lemma}	[Decomposable Structure of \tcb{$J\left(\mathbf{q};\mathbf{0}\right)$}]	\label{linearAp}
	For sufficiently large delay price $\boldsymbol{\gamma}$, the fluid value function \tcb{$J\left(\mathbf{q};\mathbf{0}\right)$} of the  base VCTS has the following decomposable structure:
	\begin{equation}   \label{linearA}
		\tcb{J\left(\mathbf{q};\mathbf{0}\right)}= \sum_{k=1}^K J_k \left(q_k \right), \quad  \mathbf{q} \in \overline{\boldsymbol{\mathcal{Q}}}
	\end{equation} 
	where  $J_k \left(q_k \right)$ is the \emph{per-flow fluid value function} for the $k$-th data flow given by:
	\begin{equation}	\label{pareJkEg2}
 \left\{
	\begin{aligned}	 
		q_k(y) &=  \frac{\lambda_k}{\gamma_k} \left(R_k  e^{-\frac{a_k}{R_k  y}}y-\lambda_k y -a_k  E_1\left(\frac{a_k}{R_k  y} \right)+ {c_k^{\infty}}\right)		 \\
		J_k(y) &=   \frac{\lambda_k}{\gamma_k} \left(\frac{ \left(R_k  y - a_k \right)}{2}ye^{-\frac{a_k}{R_k  y}} - \frac{\lambda_k}{2}y^2 \right.  \\
		& \left. +  \frac{a_k^2}{2 R_k } E_1\left(\frac{a_k}{R_k  y} \right)\right) + b_k	
	   \end{aligned}
   \right.
 	 \end{equation}
where $a_k \triangleq 2^{R_k}-1$, {$c_k^{\infty}=a_k E_1\left(\log \frac{R_k }{\lambda_k} \right)$} and $E_1(x) \triangleq \int_1^{\infty} \frac{e^{-tx}}{t}\mathrm{d}t$ is the exponential integral function.  $b_k$ is chosen to satisfy\footnote{To find $b_k$, firstly solve $q_k(y_k^0)=0$ using one-dimensional search techniques (e.g., bisection method). Then $b_k$ is chosen such that $J_k(y_k^0)=0$.} the boundary condition $J_k(0)=0$.~\hfill\IEEEQED
\end{Lemma}
\begin{proof} 
please refer to Appendix D.
\end{proof}

The following corollary summarizes the asymptotic property of the per-flow fluid value function $J_k \left( q_k \right)$ in Lemma \ref{linearAp}.
\begin{Corollary}	[Asymptotic Property of $J_k \left( q_k \right)$]   	\label{PropJk}
	\begin{align}	\label{asympoticJ_k}
		J_k \left( q_k \right) = \frac{\gamma_k}{2\lambda_k(R_k -\lambda_k)} q_k^2 + o(q_k^2), \quad \text{as } q_k \rightarrow \infty
	\end{align}~\hfill\IEEEQED
\end{Corollary}
\begin{proof} 
please refer to Appendix E.
\end{proof}

Next, we study the VCTS for small $\boldsymbol{\epsilon}$ by treating it as a perturbation of the base VCTS. Using perturbation analysis, we establish the following theorem on the approximation error between \tcb{$J\left(\mathbf{q};\boldsymbol{\epsilon} \right)$} and \tcb{$J\left(\mathbf{q};\mathbf{0}\right)$}.
\begin{Theorem}	[Approximation Error]	\label{ErrorEg2}	
	 The approximation error between between \tcb{$J\left(\mathbf{q};\boldsymbol{\epsilon} \right)$} and \tcb{$J\left(\mathbf{q};\mathbf{0}\right)$} is given by
\begin{align}	
		&\tcb{J\left(\mathbf{q};\boldsymbol{\epsilon} \right) }   =   \tcb{J\left(\mathbf{q};\mathbf{0}\right)}, \qquad \text{as } q_k \rightarrow \infty \ \forall k, \epsilon \rightarrow 0 \notag \\
		&+\sum_{k=1}^K  \sum_{j \neq k} \epsilon_k D_{kj} \left(q_k q_j\log q_j  + o\left(q_k q_j \log q_j \right) \right)+\mathcal{O}\left(\epsilon^2 \right)
	\end{align}
	where \tcb{$J\left(\mathbf{q};\mathbf{0}\right)$} is given in (\ref{linearA}), $\epsilon=\min_k{\epsilon_k}$ and $D_{kj} =\frac{\gamma_k(2^{R_k}-1)(2^{R_j}-1)}{\lambda_k (R_k  - \lambda_k)(R_j  - \lambda_j) 2^{R_k-1}\ln 2}$.~\hfill\IEEEQED
\end{Theorem}
\begin{proof} 
please refer to Appendix F.		
\end{proof}
\begin{figure}
\centering
\subfigure[CSIT error variance $\epsilon_k=0.01$. The delay prices are $\gamma_k = 10$, $\gamma_k = 20$ and $\gamma_k = 30$, which correspond to the average delay per flow being 3.3pcks, 3pcks and 2.8pcks (under optimal policy), respectively.]{
\includegraphics[width=3.49in]{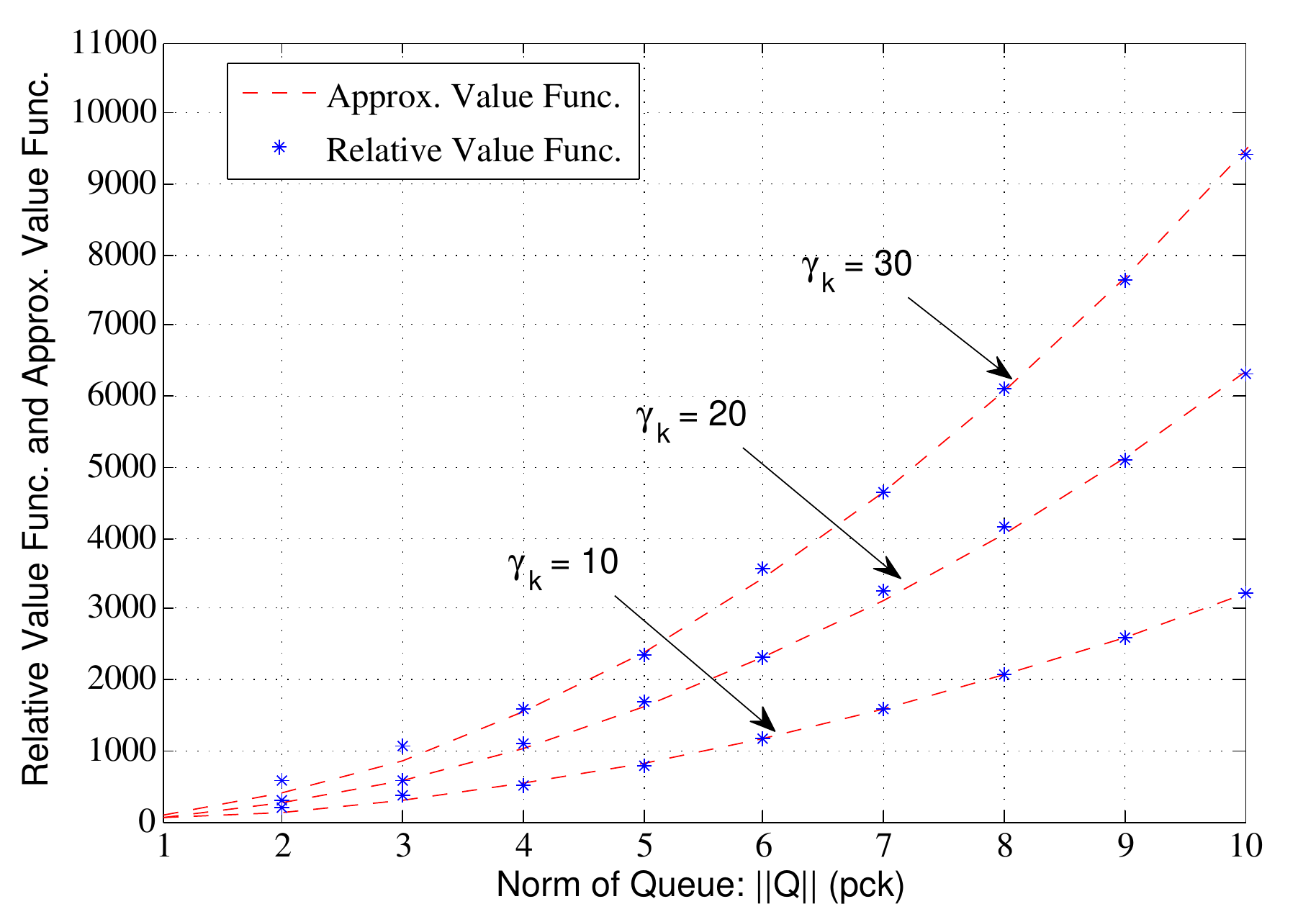}}
\hspace{0.5cm}
\subfigure[CSIT error variance $\epsilon_k=0.4$. The delay prices are $\gamma_k = 10$, $\gamma_k = 20$ and $\gamma_k = 30$, which correspond to the average delay per flow being 5.7pcks, 4.5pcks and 3.4pcks (under optimal policy), respectively.]{
\includegraphics[width=3.5in]{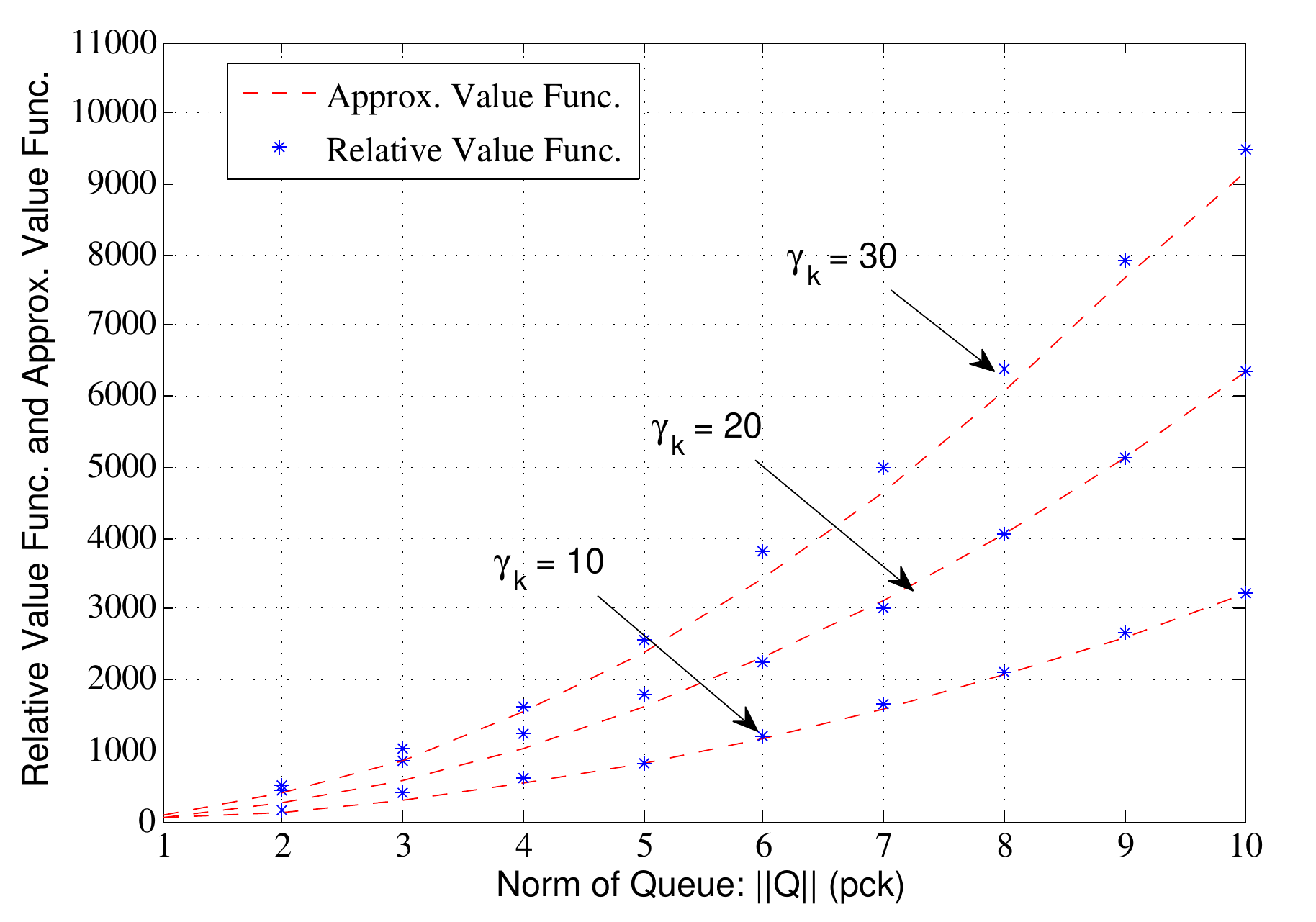}}
\caption{Relative value function $V^\ast \left(\mathbf{Q} \right)$ and approximate value function  $\widetilde{V}\left(\mathbf{Q} \right)$ versus the norm of the global queue state $\| \mathbf{Q} \|$ with $\mathbf{Q}=\{Q_1, 2, \dots, 2 \}$. The system parameters are configured as in Fig.~\ref{fig1vs} in  Section \ref{simulationP}. Note that the relative value functions are calculated using relative value iteration \cite{DP_Bertsekas}.}
\label{approx}
\end{figure}


\tcb{In practice, the CSIT error in the downlink MU-MIMO systems cannot be too large. Otherwise, the multi-user interference will severely limit the system performance of spatial multiplexing. As a result, it is important to consider the regime when $\boldsymbol{\epsilon}$ is small.}  We then obtain the following closed-form approximation of the relative value function: 
\begin{align}
	V^\ast\left(\mathbf{Q} \right) \approx \widetilde {V}\left(\mathbf{Q} \right) \triangleq \sum_{k=1}^K J_k \left(q_k \right)  +\sum_{k=1}^K \sum_{j \neq k} \epsilon_k D_{kj} q_k q_j\log q_j 	\label{perflowapp}
\end{align}
\tcb{Furthermore, based on Corollary \ref{PropJk}  and (\ref{perflowapp}), we have  $\left\{\frac{\partial \widetilde{V}\left(\mathbf{Q}\right)}{\partial Q_k}: \forall k \right\}$  are increasing functions of all $Q_k$ and $\widetilde{V}\left(\mathbf{Q}\right)=\mathcal{O}\big(\sum_{k=1}^K Q_k^2 \log Q_k \big)$.  Based on Theorem \ref{them111} and Theorem 3, the approximation error between the optimal value function $V^\ast \left(\mathbf{Q} \right)$ in Theorem \ref{LemBel} and the closed-form approximate value function  $\widetilde {V}\left(\mathbf{Q} \right)$ in (\ref{perflowapp})  is $\mathcal{O}(\epsilon)+ o(1)$. In other words, the error terms are asymptotically small w.r.t. the CSIT error variance and the slot duration.} Fig.~\ref{approx} illustrates the quality of the approximation. In the next section, we  derive a low complexity control policy using the closed-form approximate value function $\widetilde {V}\left(\mathbf{Q} \right)$ in (\ref{perflowapp}).

\begin{Remark}	[\tcb{Computational Complexity Analysis}]
\tcb{In conventional MDP/POMDP approaches, numerical methods such as value iteration or policy iteration \cite{Cao}, \cite{DP_Bertsekas} are used to obtain the relative value function, which has exponential complexity in  $K$ (where $K$  is the number of MSs). In  our framework, a closed-form approximate value function in (\ref{perflowapp}) is derived.  As a result, our solution has much lower complexity compared with the conventional brute-force value iterations, which is  illustrated in Table \ref{tabletime} in Section \ref{simulationP}.~\hfill\IEEEQED}  
\end{Remark}

\section{Low complexity Delay-Aware  Beamforming Control}   \label{KnTbased}
In this section, we use the closed-form approximate value function  in (\ref{perflowapp}) \tcb{to capture the urgency information of the $K$  data flows} and to  obtain low complexity delay-aware beamforming  control. The problem is still quite challenging because it is non-convex and there is no closed-form expression for the conditional PER. We first build a  tractable closed-form approximation of the conditional PER using \emph{Bernstein approximation} and propose a conservative formulation of the beamforming control problem. We then apply semidefinite relaxation (SDR) technique \cite{luo}, \cite{sdr2} to transform the conservative formulation into a convex problem   and propose an alternating iterative algorithm to efficiently solve the SDR problem.

\subsection{Equivalent Chance Constraint  Problem}
Using the approximate value function in (\ref{perflowapp}), we transform the problem in (\ref{conperr}) into a collection of  chance constraint problems (w.r.t. each observed state realization $\big(\hat{\mathbf{H}},\mathbf{Q}\big)$)  as shown in the following lemma:
\begin{Lemma} [Equivalent Chance Constraint Problem]	\label{equperccp}
Using the approximate value function in (\ref{perflowapp}),  the problem in (\ref{conperr}) can be transformed into the a collection of  chance constraint problems  w.r.t. each given observed state realization $\big(\hat{\mathbf{H}},\mathbf{Q} \big)$:
	\begin{subequations}	\label{ApproxPol2}
	\begin{eqnarray}	
	    & \hspace{-1cm}\underset{\mathbf{w}, \boldsymbol{\rho}}\min&   \ \sum_{k=1}^K \Big(\|\mathbf{w}_k\|^2 - \frac{\partial \widetilde {V}\left(\mathbf{Q}\right)}{\partial Q_k}  R_k    \left( 1- \rho_k\right)  \Big)	\label{objorgft} \\
		&\hspace{-1cm} \text{s.t.        }&    \  \Pr\Big[ R_k\leq C_k\left( \mathbf{H}, \mathbf{w}\right) \big| \hat{\mathbf{H}},\mathbf{Q} \Big] \geq 1- \rho_k,  \forall k	\label{percsitcons}	\\
		&\hspace{-1cm} \text{       }&    \  0 \leq \rho_k \leq 1, \quad \forall k
	\end{eqnarray}
	\end{subequations}
	where $\boldsymbol{\rho}   \triangleq  \{\rho_k: \forall k \}$ is the collection of the conditional PER targets of all the $K$ data flows.
\end{Lemma}
\begin{proof}
	please refer to Appendix G.
\end{proof}
\begin{Remark}	[Interpretation of Lemma \ref{equperccp}]
	Instead of a fixed conditional PER target $\rho_k$ for each data flow, the conditional PER obtained by  solving the problem in (\ref{ApproxPol2}) is  adaptive to the CSIT and QSI. Specifically, when the data flow is  urgent\footnote{When $Q_k$ is large, the associated weight $\frac{\partial \widetilde {V}\left(\mathbf{Q}\right)}{\partial Q_k}  $ for $\rho_k$ in (\ref{objorgft}) gets larger compared with the weights  $\frac{\partial \widetilde {V}\left(\mathbf{Q}\right)}{\partial Q_j}  $ for $j \neq k$ (since the increase of $\frac{\partial \widetilde {V}\left(\mathbf{Q}\right)}{\partial Q_j}$ when $Q_k$ is large is $\mathcal{O}(\epsilon)$). Hence, the optimized  $\rho_k$ will be small.} (i.e., $Q_k$ is large), $\rho_k$ will be smaller indicating that the system tends to be more aggressive to transmit information bits.~\hfill\IEEEQED
\end{Remark}

Note that the above chance constrained problem is difficult to solve since the conditional PER constraint in (\ref{percsitcons}) does not have closed-form expression. In the next subsection, we  obtain a  tractable closed-form  approximation of  the conditional PER constraint using Bernstein approximation \cite{Chancecon1}, \cite{fixedper}.

\subsection{Bernstein Approximation of Conditional PER Constraint}
In this part, we  obtain a closed-form  approximation of the conditional PER constraint  based on a Bernstein-type inequality \cite{Chancecon1}, \cite{fixedper}.   We first express the conditional PER constraint in (\ref{percsitcons}) in the following equivalent form:
\begin{align}
	\Pr \left[ \mathbf{v}_k\mathbf{M}_k\big(\mathbf{w} \big) \mathbf{v}_k^{\dagger}  + 2 \mathrm{Re}\left\{  \mathbf{v}_k \mathbf{z}_k\big(\mathbf{w}\big) \right\} \geq e_k\big(\mathbf{w}\big)\right] \geq 1 - \rho_k		\label{transchan}
\end{align}
where $\mathbf{v}_k$ is the normalized Gaussian random vector  in Assumption \ref{impCSITm} and  $\mathrm{Re}\left\{ \cdot\right\}$ denotes the real part of the associated argument. $\mathbf{M}_k$, $ \mathbf{z}_k$ and $ e_k$ are given as below
\begin{eqnarray}
	 &&\mathbf{M}_k\big(\mathbf{w}\big) \triangleq \epsilon_k \big(\frac{1}{2^{R_k}-1} w_kw_k^{\dagger}- \sum_{j \neq k} w_jw_j^{\dagger} \big),  \notag \\
	 &&\mathbf{z}_k\big(\mathbf{w}\big) 	\triangleq \sqrt{\epsilon_k} \big(\frac{1}{2^{R_k}-1} w_kw_k^{\dagger}- \sum_{j \neq k}  w_jw_j^{\dagger}\big) \hat{\mathbf{h}}_k^{\dagger}		\notag \\
	 &&e_k\big(\mathbf{w} \big)  \triangleq 	1-  \hat{\mathbf{h}}_k \big(\frac{1}{2^{R_k}-1}  w_kw_k^{\dagger}- \sum_{j \neq k} w_jw_j^{\dagger} \big) \hat{\mathbf{h}}_k^{\dagger} \notag 
\end{eqnarray}
The conditional PER constraint in  (\ref{transchan}) involves a quadratic form of the complex Gaussian  random variables $\{\mathbf{v}_k\}$. To find a closed-form approximation of the PER constraint based on (\ref{transchan}),  we use the following lemma:
\begin{Lemma}	[Bernstein-type Inequality]	\label{bernsteinine}
	Let $A=\mathbf{v} \mathbf{M}  \mathbf{v}^{\dagger}+ 2 \mathrm{Re}\left\{  \mathbf{v} \mathbf{z}\right\}$, where $M \in \mathbb{H}^{N_t \times N_t}$ is a complex Hermitian matrix,  $\mathbf{z}_k \in \mathbb{C}^{N_t \times 1 }$ and $\mathbf{v}_k \sim \mathcal{CN}\left(0, \textbf{I}_{N_t} \right)$. Then, for any $\delta > 0$, we have
	\begin{small}
	\begin{align}
		\Pr \left[A \geq \mathrm{Tr} \left( \mathbf{M}  \right) -\sqrt{2 \delta} \sqrt{\| \mathbf{M} \|_F^2 + 2 \| \mathbf{z}\|^2} - \delta s^+ \left(\mathbf{M} \right)\right] \geq 1- e^{-\delta}	\notag 
	\end{align}
	\end{small}where $s^+\left(\mathbf{M}  \right)= \max\{\lambda_{\text{max}}(-\mathbf{M}), 0 \}$ in which $\lambda_{\text{max}}(-\mathbf{M})$ denotes the maximum eigenvalue of matrix $-\mathbf{M}$ and $\|\cdot\|_F$ denotes the matrix Frobenius norm.~\hfill\IEEEQED
\end{Lemma}
\begin{proof}
	please refer to \cite{bernsteininequ} for details.
\end{proof}

Based on Lemma \ref{bernsteinine}, we obtain  a  conservative form of  the conditional PER constraint in (\ref{transchan}) which is summarized in the following lemma:
\begin{Lemma} [Conservative Form of  Conditional PER Constraint]
	A conservative formulation of the conditional PER constraint in (\ref{transchan}) is given by
	\begin{small}
	\begin{align}
		 \mathrm{Tr} \left( \mathbf{M}_k\big(\mathbf{w}\big)  \right) -\sqrt{2 \delta_k} \sqrt{\| \mathbf{M}_k\big(\mathbf{w}\big)\|_F^2 + 2 \| \mathbf{z}_k\big(\mathbf{w}\big) \|^2} \notag \\
		  - \delta_k s^+ \left( \mathbf{M}_k\big(\mathbf{w}\big)\right) \geq  e_k\big(\mathbf{w} \big) \label{conser}
	\end{align}
	\end{small}where $\delta_k = - \ln \left(\rho_k\right)$. In other words, the constraint in (\ref{conser}) is a sufficient condition for the conditional PER constraint in (\ref{transchan}). ~\hfill\IEEEQED
\end{Lemma}
	The conservative formulation in (\ref{conser}) provides a closed-form approximation of the PER constraint in (\ref{transchan}).  Based on (\ref{conser}), we have the following  conservative formulation of the problem in (\ref{ApproxPol2}).
	\begin{Problem}	[Conservative Formulation of (\ref{ApproxPol2})]	\label{conserform}
	\begin{subequations}
	\begin{eqnarray}	
	& \hspace{-1cm} \underset{\mathbf{w}, \boldsymbol{\delta},  \mathbf{x}, \mathbf{y}}\min&   \ \sum_{k=1}^K \Big(\|\mathbf{w}_k\|^2 - \frac{\partial \widetilde {V}\left(\mathbf{Q}\right)}{\partial Q_k}  R_k    \left( 1- \rho_k\right)  \Big)	 \\
	& \hspace{-1cm}\text{s.t.}&	\mathrm{Tr} \left( \mathbf{M}_k\big(\mathbf{w}\big) \right) -\sqrt{2 \delta_k} x_k - \delta_k y_k\geq e_k\big(\mathbf{w} \big),  \forall k\\
	& \hspace{-1cm}\hspace{1cm}&\sqrt{\| \mathbf{M}_k\big(\mathbf{w}\big) \|_F^2 + 2 \| \mathbf{z}_k\big(\mathbf{w}\big) \|^2} \leq x_k, \quad \forall k	\\
	& \hspace{-1cm}\hspace{1cm}&y_k \mathbf{I}_{N_t} + \mathbf{M}_k\big(\mathbf{w}\big)   \succeq \mathrm{0}	, \quad \forall k	\\
	& \hspace{-1cm}\hspace{1cm}&y_k \geq 0, \delta_k>0, \quad \forall k
	\end{eqnarray}
	\end{subequations}
	where we denote $\boldsymbol{\delta} \triangleq \{\delta_k=-\ln \left(\rho_k\right): \forall k \}$. $\mathbf{x} \triangleq \{x_k \in \mathbb{R}: \forall k\}$ and $\mathbf{y} \triangleq \{y_k \in \mathbb{R}: \forall k\}$ are slack variables.~\hfill\IEEEQED
	\end{Problem}

Problem \ref{conserform} is still non-convex due to the fact that $ \mathbf{M}_k\big(\mathbf{w}\big)$, $ \mathbf{z}_k\big(\mathbf{w}\big) $ and $e_k\big(\mathbf{w}\big)$ are indefinite quadratic in $\mathbf{w}$ \cite{luo}. To efficiently solve the problem, we adopt the semidefinite relaxation (SDR) technique \cite{luo}, \cite{sdr2}.  Define  $W_k=\mathbf{w}_k \mathbf{w}_k^{\dagger}$. Thus, we have $W_k  \succeq \mathbf{0}$ and $\mathrm{rank}\big(W_k  \big)=1$.  Removing the rank-one constraint on $W_k$, we have the following SDR of Problem \ref{conserform}.
\begin{Problem} [SDR of Problem \ref{conserform}]	\label{relaxed}
\begin{subequations}	\label{sdrprobbb}
\begin{eqnarray}	 
	 &\hspace{-1.3cm}\underset{\mathbf{W},  \boldsymbol{\delta}, \mathbf{x},\mathbf{y}}{\min}&   \sum_{k=1}^K  \Big(  \mathrm{Tr}\left(W_k\right) - \frac{\partial \widetilde {V}\left(\mathbf{Q}\right)}{\partial Q_k}  R_k  \big(1-e^{-\delta_k}\big)  \Big) 		\label{objweit} \\
	&\hspace{-1.5cm}\text{s.t.}& 	  \mathrm{Tr} \left( \mathbf{M}_k\big(\mathbf{W} \big) \right) -\sqrt{2 \delta_k} x_k - \delta_k y_k\geq e_k\big(\mathbf{W}  \big),  \forall k	\label{linearp4} \\
	&\hspace{-1.5cm} \hspace{1cm} & \sqrt{\| \mathbf{M}_k\big(\mathbf{W} \big) \|_F^2 + 2 \| \mathbf{z}_k\big(\mathbf{W} \big) \|^2} \leq x_k , \quad \forall k \label{secconp4} \\
	& \hspace{-1.5cm}\hspace{1cm} &y_k \mathbf{I}_{N_t} + \mathbf{M}_k\big(\mathbf{W} \big)   \succeq \mathrm{0}, \quad \forall k	\label{psdp41} \\
	&\hspace{-1.5cm} \hspace{1cm} & y_k \geq 0 , W_k  \succeq \mathrm{0},\delta_k \geq 0	, \quad \forall k	\label{psdp42}
\end{eqnarray}
\end{subequations}
	where  $\mathbf{W}  \triangleq \big\{W_k: \forall k \big\}$, $\mathbf{M}_k\big(\mathbf{W} \big) \triangleq \epsilon \Big(\frac{1}{2^{R_k}-1} W_k - \sum_{j \neq k}W_j   \Big)$, $\mathbf{z}_k\big(\mathbf{W} \big) \triangleq \sqrt{\epsilon_k} \Big(\frac{1}{2^{R_k}-1}W_k  - \sum_{j \neq k} W_j  \Big) \hat{\mathbf{h}}_k^{\dagger}$ and $e_k\big(\mathbf{W}  \big)  \triangleq 	1-  \hat{\mathbf{h}}_k \Big(\frac{1}{2^{R_k}-1} W_k  - \sum_{j \neq k} W_j  \Big) \hat{\mathbf{h}}_k^{\dagger}$.~\hfill\IEEEQED
\end{Problem}

For Problem \ref{relaxed}, the objective function in (\ref{objweit}), the second order cone constraint in \tcb{(\ref{secconp4})},  the positive semidefinite constraints in (\ref{psdp41}) and the linear constraints in (\ref{psdp42}) are all convex. However, it is difficult to check  the convexity of Problem \ref{relaxed} due to the bilinear terms $\sqrt{2 \delta_k} x_k $ and $\delta_k y_k$ in (\ref{linearp4}). In the following lemma, we show that Problem \ref{relaxed} is  convex.
\begin{Lemma}	[Convexity of Problem \ref{relaxed}]	\label{asycvx}
	Problem \ref{relaxed} is a convex optimization problem.
\end{Lemma}
\begin{proof}
	please refer to Appendix H.
\end{proof}

Let $\left\{ \mathbf{W}^{\ast} , \boldsymbol{\delta}^{\ast}, \mathbf{x}^{\ast}, \mathbf{y}^\ast \right\}$ be the optimal  point of Problem \ref{relaxed}, where $\mathbf{W}^{\ast}=\{W_k^{\ast}:\forall k\}$. If  $W_k^{\ast}$ is of rank one, we can write it as  $W_k^{\ast}=\mathbf{w}_k^{\ast}\big(\mathbf{w}_k^{\ast} \big)^{\dagger}$ and then $\mathbf{w}_k^{\ast}$ is the solution.  Otherwise, we apply standard rank reduction techniques (such as the Gaussian Randomization Procedure (GRP) \cite{luo}) to obtain a rank-one approximate solution from $\mathbf{W}^{\ast}$. 

We propose the following alternating iterative algorithm   to efficiently  solve Problem \ref{relaxed}.  
\begin{Algorithm} [Alternating Iterative Algorithm to Problem \ref{relaxed}]		\label{BFalg} \
	\begin{itemize}
		\item \textbf{Step 1 [Initialization]:}  Set $n=0$ and choose $\boldsymbol{\delta}(0) \succeq \mathbf{0}$. 
		\item \textbf{Step 2 [Update on $\mathbf{W} $]:} Based on  $\boldsymbol{\delta}(n)$, obtain the optimal solution $\mathbf{W}  (n)$ of the  following convex conic program using the cutting plain or ellipsoid method \cite{boyd}:
	\begin{small}
	\begin{subequations}	\label{alt1}
	\begin{eqnarray}	
	&\hspace{-2.2cm} \underset{\mathbf{W},  \mathbf{x},\mathbf{y}}{\min}&   \sum_{k=1}^K \mathrm{Tr}\left(W_k\right) 	 \\
	&\hspace{-2.2cm}\text{s.t.}& 	  \mathrm{Tr} \left( \mathbf{M}_k\big(\mathbf{W} \big) \right) -\sqrt{2 \delta_k(n)} x_k - \delta_k(n) y_k\geq e_k\big(\mathbf{W}  \big),  \forall k	\\
	&\hspace{-1.8cm} \hspace{0.5cm} & \sqrt{\| \mathbf{M}_k\big(\mathbf{W} \big) \|_F^2 + 2 \| \mathbf{z}_k\big(\mathbf{W} \big) \|^2} \leq x_k, \quad \forall k   \\
	&\hspace{-1.8cm} \hspace{0.5cm} &y_k \mathbf{I}_{N_t} + \mathbf{M}_k\big(\mathbf{W} \big)   \succeq \mathrm{0}, \quad \forall k	 \\
	&\hspace{-1.8cm} \hspace{0.5cm} & y_k \geq 0 , \ W_k  \succeq \mathbf{0}	, \quad \forall k 
	\end{eqnarray}
	\end{subequations}\end{small}
		\item \textbf{Step 3 [Update on $\boldsymbol{\delta}$]:} Based on $\left\{\mathbf{W}  (n), \mathbf{x}(n), \mathbf{y}(n) \right\}$, obtain the optimal solution $\boldsymbol{\delta}(n+1)$  of the  following quadratically constrained  program using the interior-point method \cite{boyd}:
		\begin{small}
		\begin{subequations}	\label{alt2}
		\begin{eqnarray}	
 &&\hspace{-1.6cm}\underset{\boldsymbol{\delta}}{\min}   \sum_{k=1}^K  \frac{\partial \widetilde {V}\left(\mathbf{Q}\right)}{\partial Q_k} R_k\tau  e^{-\delta_k}  	\label{objdelta} \\
	&&\hspace{-1.6cm} \text{s.t.} 	\hspace{0.2cm}  \mathrm{Tr} \left( \mathbf{M}_k\big(\mathbf{W} (n)\big) \right) -\sqrt{2 \delta_k} x_k(n) - \delta_k y_k(n)\notag \\
	&&\hspace{3cm}\geq e_k\big(\mathbf{W} (n) \big),  \forall k	\\
	&&  \hspace{-1cm}\delta_k \geq 0 ,  \forall k 
		\end{eqnarray}
		\end{subequations}
		\end{small}
		\item \textbf{Step 4 [Termination]:}  Set $n = n + 1$ and go to Step 2 until a certain termination condition is satisfied.~\hfill\IEEEQED
	\end{itemize}	
\end{Algorithm}

\begin{Remark}	[Convergence Property of Algorithm \ref{BFalg}]
	In Algorithm \ref{BFalg}, $\mathbf{W} $ and $\boldsymbol{\delta}$ are optimized alternatively in the problems in (\ref{alt1}) and (\ref{alt2}). Since 
the value of the objective function after each iteration is monotonically nonincreasing and is bounded below\footnote{Since $e^{-\delta_k}>0$ and $W_k  \succeq \mathrm{0}$ for all $k$, then $\sum_{k=1}^K  \left(  \mathrm{Tr}\left(W_k\right) - \frac{\partial \widetilde {V}\left(\mathbf{Q}\right)}{\partial Q_k}R_k  \big(1-e^{-\delta_k}\big) \right)  > -\sum_{k=1}^K  \frac{\partial \widetilde {V}\left(\mathbf{Q}\right)}{\partial Q_k} R_k$.},  Algorithm \ref{BFalg} must converge to a stationary point\footnote{A stationary point of Problem \ref{relaxed} satisfies the associated KKT conditions.} of Problem \ref{relaxed}. In addition, according to the convexity of Problem \ref{relaxed} in Lemma \ref{asycvx}, any converged stationary point is also a global optimal point.~\hfill\IEEEQED
\end{Remark}

Fig.~\ref{fc} summarizes the  overall  delay-constrained beamforming solution. \tcb{Note that the delay-awareness of the beamforming solution is embraced via the approximate value function $ \widetilde {V}\left(\mathbf{Q}\right)$  in (\ref{objdelta}), which gives the urgency information about the $K$ data flows. The complexity of the beamforming solution comes from the calculation of  value functions and the optimization of the per-stage problem. Due to the closed-form approximation of the value function, the overall complexity of the solution is just the complexity of solving the SDR in (\ref{sdrprobbb}), which is polynomial in $K$ \cite{boyd}.}

\begin{figure}[!htbp]
  \centering
  \includegraphics[width=2.8in]{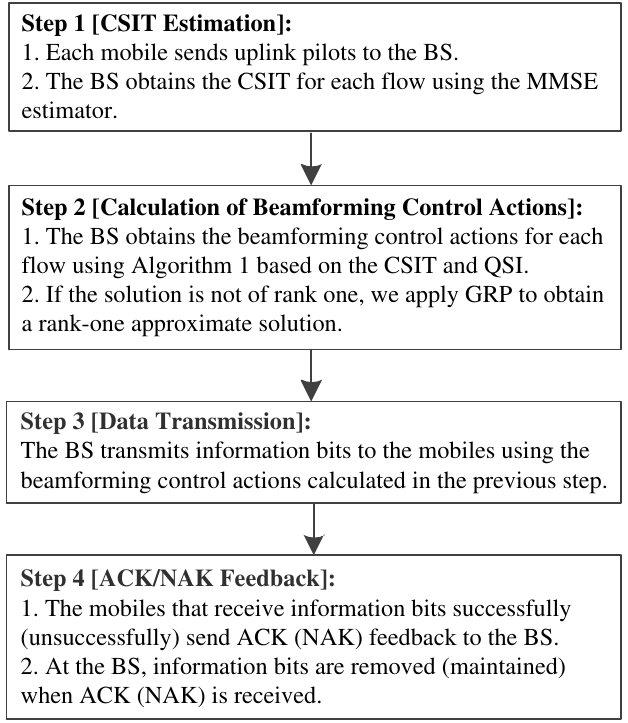}
  \caption{Flow chart of the overall delay-constrained beamforming solution.}
  \label{fc}
\end{figure}

\section{Simulation Results and Discussions}	\label{simulationP}
In this section, we compare the  performance  of the proposed beamforming control scheme in Algorithm \ref{BFalg}  with the following three baseline schemes using numerical simulations.
\begin{itemize} 
	\item \textbf{Baseline 1, Random Beamforming (RB) Scheme \cite{RB_scheme}:}   The BS applies  random beamforming  control and uniform power allocation for  each data flow.  
	\item \textbf{Baseline 2, Fixed PER Beamforming (FPB) Scheme \cite{fixedper}:} This scheme targets solving the problem in (\ref{ApproxPol2}) with fixed conditional PER  ($\rho_k=0.1$ for all $k$) for each data flow. We apply similar techniques (Bernstein approximation, SDR) to obtain the solution of the corresponding  control problem.
	\item \textbf{Baseline 3, CSIT-Adaptive PER Beamforming (CAPB) Scheme:}  This scheme is an extension of the method presented in \cite{fixedper}, which minimizes $\sum_{k=1}^K \big(\|\mathbf{w}_k\|^2 + \beta  \rho_k   \big)$. $\beta$ measures the relative importance between the power cost of the BS $\sum_{k=1}^K \|\mathbf{w}_k\|^2 $ and the sum of the per-flow PERs $\sum_{k=1}^K \rho_k$.
\end{itemize}

\begin{figure}[h]
\centering 
{\includegraphics[width=3.5in]{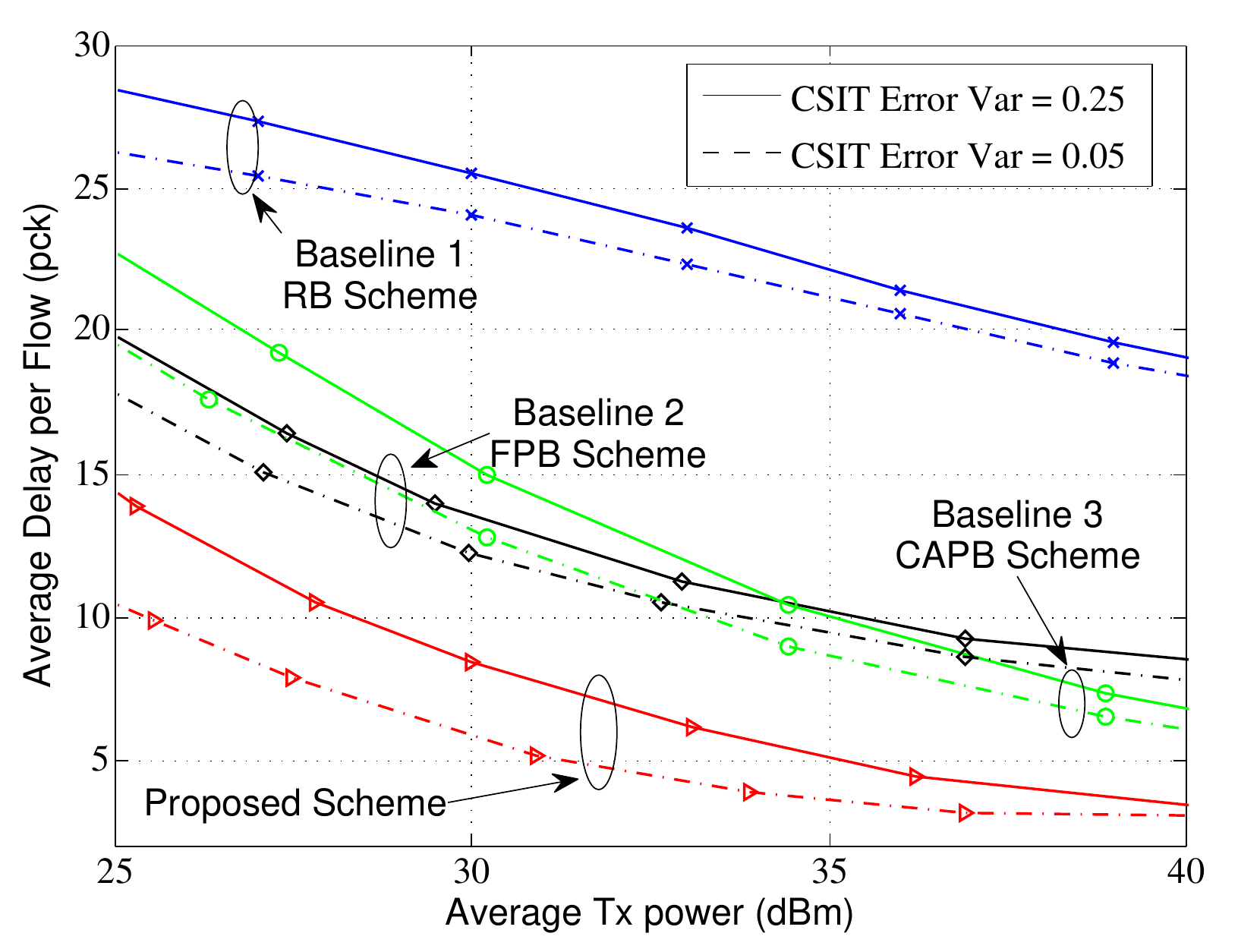}}	
\caption{Average delay per flow versus average transmit power at $\epsilon_k=0.05$ and $\epsilon_k=0.25$. The number of mobiles is $K=5$ and the average data arrival rate is $\lambda_k=0.8$pck/slot.}
\label{fig1vs}
\end{figure}

\begin{figure}[h]
\centering
{\includegraphics[width=3.5in]{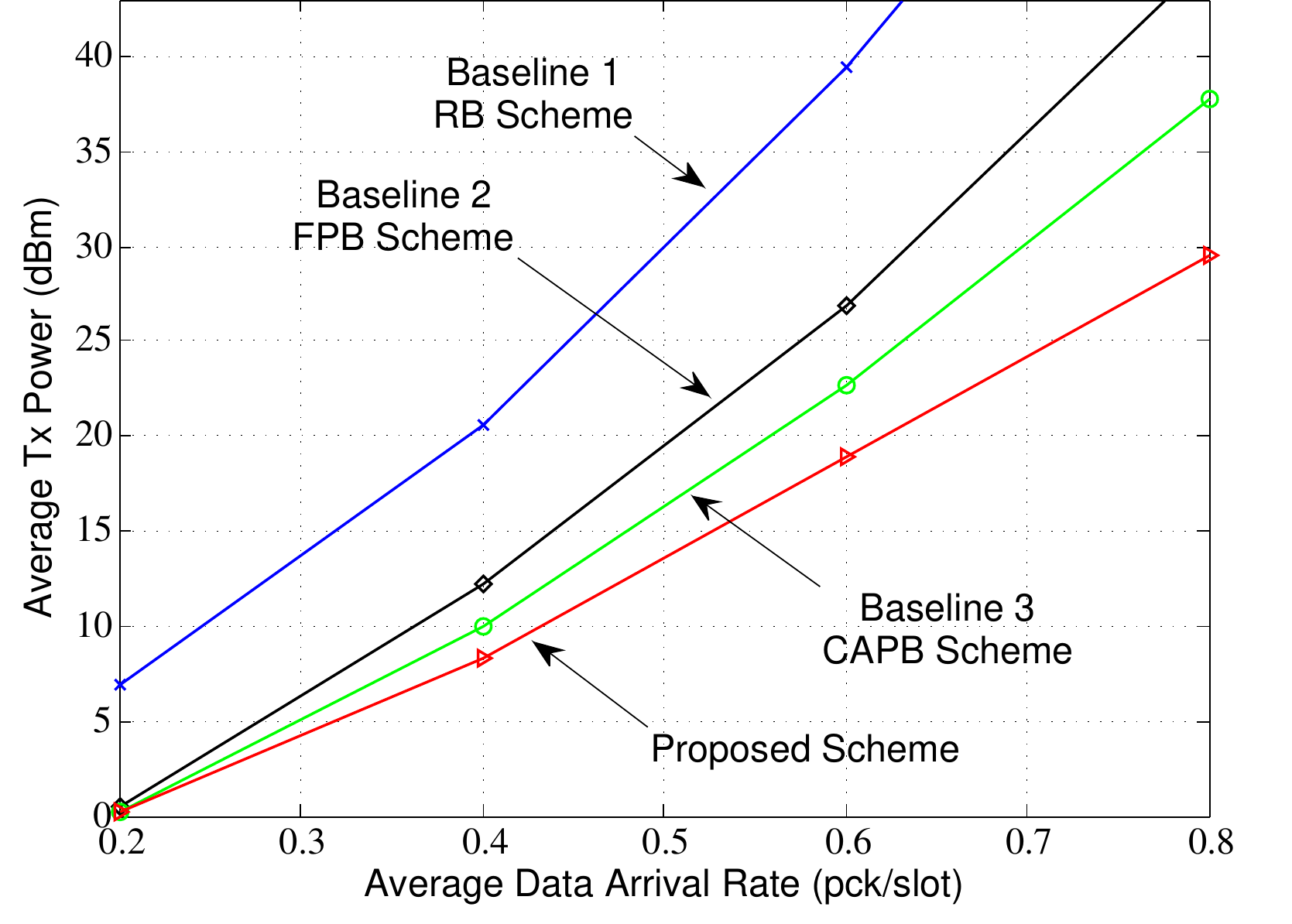}}
\caption{Average transmit power versus average data arrival rate with per flow average delay requirement  being 8pcks. The CSIT error variance is $\epsilon_k=0.1$. The number of mobiles is $K=6$.}
\label{fig2vss}
\end{figure}

\begin{figure}[h]
\centering
{\includegraphics[width=3.5in]{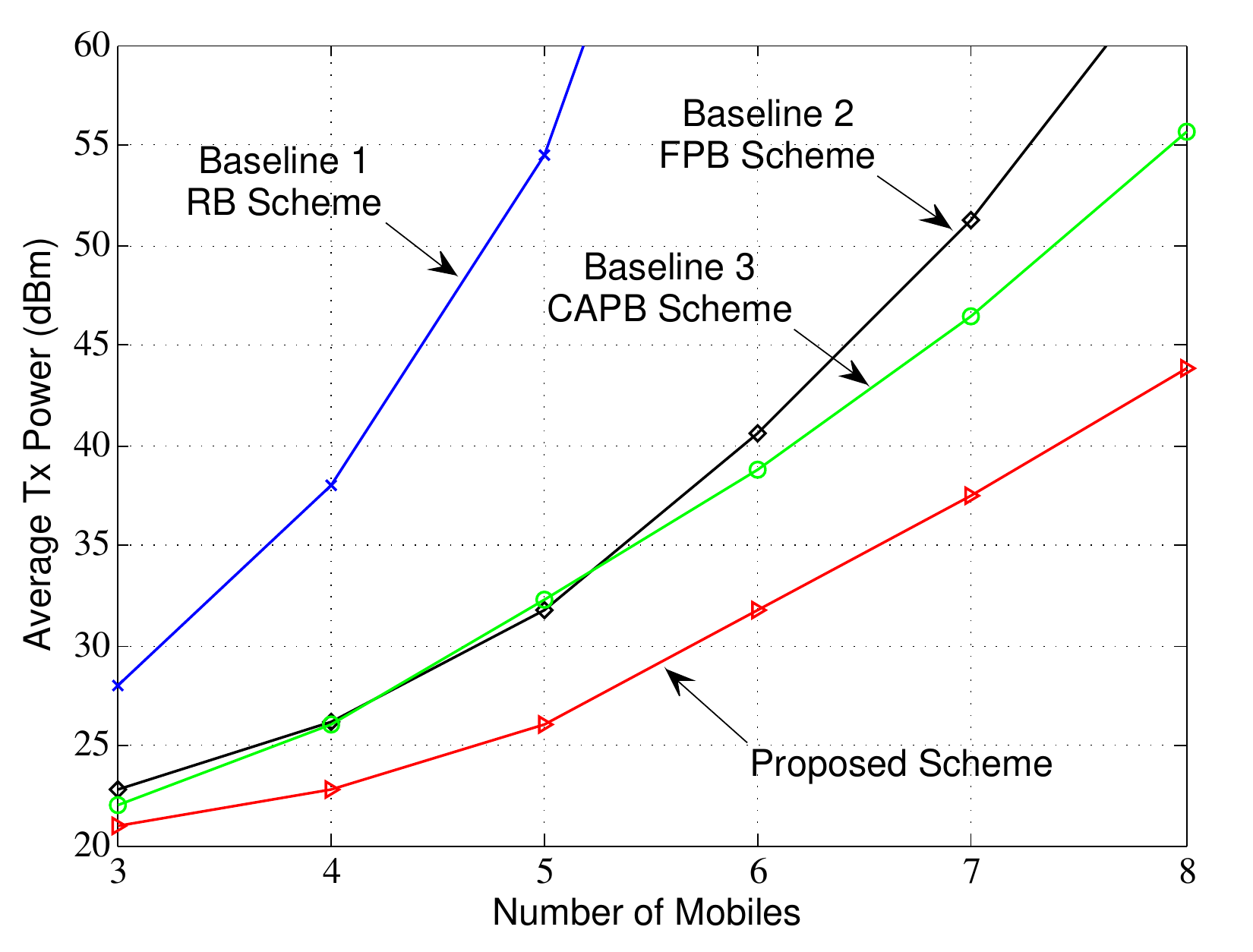}}
\caption{Average transmit power  versus number of mobiles with per flow  average delay requirement  being 12pcks. The CSIT error variance is $\epsilon_k=0.1$. The average data arrival rate is $\lambda_k=0.8$pck/slot.}
\label{fig2vssss}
\end{figure}

In the simulations, we consider a downlink MU-MIMO system where the number of transmit  antennas of the BS is equal to the number of mobiles. The complex fading coefficient and the channel noise are $\mathcal{CN}(0,1)$ distributed. We consider Poisson packet arrival with average arrival rate $\lambda_k$ (pcks/slot) and  deterministic packet size $\overline{N}_k=15$Kbits. The decision slot duration $\tau$ is $5$ms. The total bandwidth is  $BW=10$MHz. We consider the CSIT error model  in Assumption \ref{impCSITm} with CSIT error variance\footnote{\tcb{Note that to support spatial multiplexing in MU-MIMO, reasonable CSIT quality at the BS is required. As such, we consider $\epsilon \in (0,0.5]$ in the simulations.}} \tcb{$\epsilon \in (0,0.5]$ \cite{noisevar}}. The constant packet departure rate is $\frac{R_k\tau BW}{\overline{N}_k}=1$pck/slot.  Furthermore, $\gamma_k$ is the same for all $k$.

Fig.~\ref{fig1vs} illustrates the average delay per flow versus the average transmit power of the BS. The average delay of all the schemes decreases as the average transmit power increases. It can be observed that there is significant performance gain of the proposed scheme compared with all the baselines.  This gain is contributed by the CSIT and QSI aware dynamic beamforming control.

Fig.~\ref{fig2vss}  illustrates the average transmit power versus the average data arrival rate with  per flow average delay requirement  being 8pcks. The average transmit power of all the schemes increases as the average data arrival rate increases.   It can be observed that there is significant performance gain of the proposed scheme compared with all the baselines across a wide range of the average data arrival rates. 

Fig.~\ref{fig2vssss}  illustrates the average transmit power versus the number of  mobiles with per flow average delay requirement  being 12pcks. The average transmit power of all the schemes increases as the number of mobiles increases.   This is due to the increase of the total interference for each data flow. It can be observed that there is significant performance gain of the proposed scheme compared with all the baselines across a wide range of the numbers of mobiles.

\tcb{Table \ref{tabletime}  illustrates the comparison of the MATLAB computational time of the proposed solution, the  baselines and the brute-force value iteration algorithm \cite{DP_Bertsekas}. Note that the computational time of the FPB scheme is smaller than those of the CAPB scheme and our proposed scheme in all three case (different $K$ scenarios). The reason is that FPB scheme uses fixed PER and there is no PER optimization step involved. In addition,  the computational time of  our proposed scheme is very close to that of the CAPB scheme, and the  value iteration algorithm requires very long time to converge.  This is because the proposed scheme makes use of the closed-form approximate value function  and the complexity is just the complexity of solving an SDP. However, the value iteration algorithm requires both computation of  the value functions numerically  and solving the per-stage optimization problem.}

\begin{table}
	\centering
\begin{tabular}{|c|c|c|c|}
	\hline
		  & $\tcb{K=4}$  & $\tcb{K=6}$ & $\tcb{K=8}$    \\
	\hline
		\tcb{Baseline 1, RB Scheme} & \tcb{$<1$ms} & \tcb{$<1$ms} & \tcb{$<1$ms}    			\\
		\tcb{Baseline 2, FPB Scheme} &  \tcb{1.524s} & \tcb{1.811s} &   \tcb{2.951s}	 \\
		\tcb{Baseline 3, CAPB Scheme} & \tcb{2.271s} & \tcb{3.024s}  &   \tcb{4.585s}		\\
	           \tcb{Proposed Scheme} & \tcb{2.316s} & \tcb{3.093s}  &  \tcb{4.676s}	\\	
	           \tcb{Value Iteration Algorithm} &   \tcb{$>10^3$s} & \tcb{$>10^3$s}  &  \tcb{$>10^3$s} 	\\
	\hline
\end{tabular}
	\caption{\tcb{Comparison of the MATLAB computational time of the proposed scheme, the baselines and  the value iteration algorithm in one decision slot. The computational time of the  value iteration algorithm for different $K$ are  all greater than $10^3$s. The CSIT error variance is $\epsilon_k=0.1$. The average data arrival rate is $\lambda_k=0.8$pck/slot.}}	
	\label{tabletime}
\end{table}

\section{summary}
In this paper, we propose a low complexity delay-constrained beamforming control for downlink MU-MIMO systems with imperfect CSIT. We show that the delay-constrained control problem can be modeled as a POMDP. We first introduce the VCTS and derive a closed-form approximate value function using perturbation theory. We then build a tractable closed-form approximation of the conditional PER    using Bernstein approximation.  Based on the two approximations, we propose a conservative formulation of the original DT-POMDP problem and propose an alternating iterative algorithm to efficiently solve the associated SDR problem. Numerical results show that the proposed beamforming control scheme has much better  performance than the other three baselines.

\section*{Appendix A: Proof of Theorem \ref{LemBel}}
Following \emph{Proposition 4.6.1} of \cite{DP_Bertsekas},  the sufficient conditions for optimality of Problem \ref{IHAC_MDP} is that there exists a ($\eta^\ast, \{ V^\ast\left(\mathbf{Q} \right) \}$) that satisfies the following Bellman equation and $V^\ast$ satisfies the transversality condition in (\ref{transodts})  for all  admissible control policy $\Omega$ and initial  state $\mathbf{Q}\left(0 \right)$:
\begin{small}\begin{align}
	& \theta {\tau} + V^\ast\left(\boldsymbol{\chi} \right) = \min_{ \Omega (\hat{\mathbf{H}},\mathbf{Q})} \Big[ c\big(\mathbf{Q}, \Omega\big(\hat{\mathbf{H}},\mathbf{Q}\big)\big){\tau}  \\
	& +  \sum_{\boldsymbol{\chi}', \hat{\boldsymbol{\chi}}' } \Pr\big[ \boldsymbol{\chi}' \big| \boldsymbol{\chi}, \Omega\big( \hat{\mathbf{H}}, \mathbf{Q}\big)\big]  V^\ast\left(\boldsymbol{\chi}'\right)    \Big]	= \min_{ \Omega (\hat{\mathbf{H}},\mathbf{Q})}\Big[ c\big(\mathbf{Q}, \Omega\big(\hat{\mathbf{H}},\mathbf{Q}\big)\big) {\tau} \notag \\
	& + \sum_{\mathbf{Q}'} \sum_{\hat{\mathbf{H}}', \mathbf{H}' }  \Pr \big[ \mathbf{Q}'\big| \mathbf{H}, \mathbf{Q}, \Omega\big(\hat{\mathbf{H}}, \mathbf{Q}  \big) \big] \Pr \big[\hat{\mathbf{H}}', \mathbf{H}' \big]  V^\ast\left(\boldsymbol{\chi}' \right)    \Big]	\notag
\end{align}\end{small}Taking expectation w.r.t. $\hat{\mathbf{H}}', \mathbf{H}'$ on both sizes of the above equation, we have
\begin{align}
	& \theta {\tau} + V^\ast\left(\mathbf{Q} \right)  = \min_{\Omega\left(\hat{\boldsymbol{\chi}} \right)} \mathbb{E}\Big[ c\big(\mathbf{Q}, \Omega\big(\hat{\mathbf{H}},\mathbf{Q}\big)\big))  {\tau} 	\notag \\
	& + \sum_{\mathbf{Q}'}   \Pr \big[ \mathbf{Q}'\big| \mathbf{H}, \mathbf{Q}, \Omega\big(\hat{\mathbf{H}}, \mathbf{Q}  \big) \big]   V^\ast\left(\mathbf{Q}' \right)   \Big| \mathbf{Q}  \Big]		 \\
	&= \min_{\Omega\left(\mathbf{Q} \right)} \Big[ \widetilde{c}\left( \mathbf{Q}, \Omega\left(\mathbf{Q} \right)\right) {\tau} +  \sum_{\mathbf{Q}'}\Pr\left[ \mathbf{Q}'| \mathbf{Q},  \Omega\left(\mathbf{Q} \right)\right]V^\ast\left(\mathbf{Q} '\right)    \Big]\notag
\end{align}
where we denote $V^\ast\left(\mathbf{Q} \right) = \mathbb{E}\big[V^\ast\left(\boldsymbol{\chi}', \hat{\boldsymbol{\chi}}' \right) \big| \mathbf{Q}\big]$, $\widetilde{c}\left( \mathbf{Q}, \Omega\left(\mathbf{Q} \right)\right) =\mathbb{E}\big[ c\big(\mathbf{Q}, \Omega\big(\hat{\mathbf{H}},\mathbf{Q}\big)\big)\big| \mathbf{Q}  \big]$  and  $\Pr\left[ \mathbf{Q}'| \mathbf{Q},  \Omega\left(\mathbf{Q} \right)\right] \\ =  \mathbb{E}\big[\Pr \big[ \mathbf{Q}'\big| \mathbf{H}, \mathbf{Q}, \Omega\big(\hat{\mathbf{H}}, \mathbf{Q}  \big) \big]  \big| \mathbf{Q} \big]$. Therefore, we obtain the equivalent Bellman equation in (\ref{OrgBel}) in Theorem \ref{LemBel}.

\section*{{Appendix B: Proof of Corollary \ref{cor1}}}
Let $\mathbf{Q}' =(Q_1',\cdots, Q_k')= \mathbf{Q}(t+1)$ and $\mathbf{Q}=(Q_1,\cdots, Q_k)=\mathbf{Q}(t)$. For the queue dynamics in (\ref{queue_sys}) and sufficiently small $\tau$, we have
\begin{align}
	\hspace{-1cm} Q_k'  = Q_k- G_k\left(\mathbf{H},\mathbf{w}\right)  + A_k\tau,  \forall Q_k >0, k = 1,\dots, K
\end{align}
If  $V\left(\mathbf{Q}\right)$ is of class $\mathcal{C}^2(\mathbb{R}_+^K)$, we have the following Taylor expansion on $V\left( \mathbf{Q}'\right)$ in (\ref{OrgBel}):
\begin{small}\begin{align}	
	&\mathbb{E}\left[ V\left( \mathbf{Q}'\right) \big| \mathbf{Q} \right] \overset{(a)}=V\left( \mathbf{Q}\right)+\sum_{k=1}^K  \frac{\partial V\left(\mathbf{Q}\right)}{\partial Q_k} \left[  \lambda_k \right. \notag \\
	&\left.- \mathbb{E}\left[R_k\left(1-\Pr\big[R_k> C_k\left(\mathbf{H},\mathbf{w} \right) \big| \hat{\mathbf{H}},\mathbf{Q} \big]\right) \big| \mathbf{Q}\right]   \right]\tau + o(\tau)	\notag 
\end{align}\end{small}where  (a) is due to $\mathbb{E}\left[G_k\left(\mathbf{H}, \mathbf{w} \right) \big| \mathbf{Q}\right]  =\mathbb{E} \big[R_k \Pr\big[R_k\leq C_k\left(\mathbf{H},\mathbf{w} \right) \big| \hat{\mathbf{H}},\mathbf{Q} \big]\big| \mathbf{Q}  \big]$.  For notation convenience, let $F_{\mathbf{Q}}(\theta, V, \Omega(\mathbf{Q}))$ denote the \emph{Bellman operator}:
\begin{align}
	  & F_{\mathbf{Q}}(\theta, V, \Omega(\mathbf{Q})) =-\theta +  \widetilde{c}\left(\mathbf{Q}, \Omega\left(\mathbf{Q} \right)\right)  +  \sum_{k=1}^K \frac{\partial V \left(\mathbf{Q} \right) }{\partial Q_k} \left[  \lambda_k \right. \notag \\
	 &  \left. - \mathbb{E}\left[R_k\left(1-\Pr\big[R_k> C_k\left(\mathbf{H},\Omega(\hat{\mathbf{H}},\mathbf{Q} ) \right) \big| \hat{\mathbf{H}},\mathbf{Q}  \big]\right)  \big| \mathbf{Q}\right]   \right] \notag \\
	& +\nu  G_{\mathbf{Q}}(V,\Omega(\mathbf{Q})) \notag
\end{align}
for some smooth function $G_{\mathbf{Q}}$ and $\nu=o(1)$ (which asymptotically goes to zero as $\tau$ goes to zero). Denote 
\begin{align}
	F_{\mathbf{Q}}(\theta, V)=\min_{ \Omega\left( \mathbf{Q} \right)} F_{\mathbf{Q}}(\theta, V, \Omega(\mathbf{Q}))
\end{align}
Suppose $\left(\theta^\ast, V^\ast\right)$ satisfies the Bellman equation in (\ref{OrgBel}), we have
\begin{align}	\label{uniquesol}
	F_{\mathbf{Q}}\left( \theta^\ast, V^\ast\right) = \mathbf{0}, \quad \forall \mathbf{Q} \in \boldsymbol{\mathcal{Q}}
\end{align}

Similarly, if $\left(\theta, V\right)$ satisfies the approximate Bellman equation in (\ref{conperr}), we have
\begin{align}	\label{defapp}
	F^\dagger_{\mathbf{Q}}\left( \theta, V\right) = \mathbf{0}, \quad \forall \mathbf{Q} \in \boldsymbol{\mathcal{Q}}
\end{align}
where  $F^\dagger_{\mathbf{Q}}(\theta, V)=\min_{ \Omega\left( \mathbf{Q} \right)} F^\dagger_{\mathbf{Q}}(\theta, V, \Omega(\mathbf{Q}))$
and
\begin{align}
	  & F^\dagger_{\mathbf{Q}}(\theta, V, \Omega(\mathbf{Q})) \\
	=& -\theta +  \widetilde{c}\left(\mathbf{Q}, \Omega\left(\mathbf{Q} \right)\right)  +  \sum_{k=1}^K \frac{\partial V \left(\mathbf{Q} \right) }{\partial Q_k} \left[  \lambda_k - \mathbb{E}\left[R_k\left(1  \right.\right. \right. \notag \\
	&  \left.\left. \left. -\Pr\big[R_k> C_k\left(\mathbf{H},\Omega(\hat{\mathbf{H}},\mathbf{Q} ) \right) \big| \hat{\mathbf{H}},\mathbf{Q}  \big]\right)  \big| \mathbf{Q}\right]   \right]  	\notag
\end{align}

We make the following claim on the relationship between the approximate Bellman equation and the original bellman equation:
\begin{Claim}	[Relationship between (\ref{simbelman}) and (\ref{OrgBel})]	\label{claimmms}
	If $\left(\theta, V\right)$ satisfies the approximate Bellman equation in (\ref{OrgBel}), then we have $|F_\mathbf{Q}(\theta, V)| = o(1)$ for any $\mathbf{Q} \in \boldsymbol{\mathcal{Q}}$.~\hfill\IEEEQED
\end{Claim}
\begin{proof}	[Proof of Claim \ref{applemma}]
	We first have the following lemma regarding the perturbation of the optimal objective  value due to perturbation of the objective function. 
\begin{Lemma}		\label{applemma}
	Consider the following two  optimization problems:
	\begin{align}
		\mathcal{P}_1(\epsilon) =\min_{\mathbf{x}} \left[f\left(\mathbf{x}\right)+ \epsilon g\left(\mathbf{x}\right)\right]	\qquad \qquad 	\mathcal{P}_2= \min_{\mathbf{x} } f\left(\mathbf{x}\right)	\label{p1p2lamma}
	\end{align}
	for a vector variable $\mathbf{x}$ and $\mathcal{P}_1$ is a perturbed problem w.r.t. $\mathcal{P}_2$. If $\mathcal{P}_2$, $\min_{\mathbf{x}} g\left(\mathbf{x}\right)$  and $g(\mathbf{x}^\ast)$ are bounded where $\mathbf{x}^\ast=\arg\min_{\mathbf{x}}f\left(\mathbf{x}\right)$, then 
	\begin{align}
		|\mathcal{P}_1(\epsilon) - \mathcal{P}_2| =\mathcal{O}(\epsilon) \label{proofxstar222}
	\end{align}
	for sufficiently small $\epsilon$.~\hfill\IEEEQED
\end{Lemma} 
\begin{proof} [Proof of Lemma \ref{applemma}]
	For $\mathcal{P}_1(\epsilon)$, we have
	\begin{align}	\label{xxx1}
		\mathcal{P}_1(\epsilon) \geq \min_{\mathbf{x}}  f\left(\mathbf{x}\right)+ \epsilon \min_{\mathbf{x}}  g\left(\mathbf{x}\right)
	\end{align}
	on the other hand, we have
	\begin{align}	\label{xxx2}
		\mathcal{P}_1(\epsilon) \leq \min_{\mathbf{x}}  f\left(\mathbf{x}\right)+ \epsilon  g\left(\mathbf{x}^\ast\right)
	\end{align}
	where $\mathbf{x}^\ast=\arg\min_{\mathbf{x}}f\left(\mathbf{x}\right)$. Hence, if $\mathcal{P}_2$, $\min_{\mathbf{x}} g\left(\mathbf{x}\right)$  and $g(\mathbf{x}^\ast)$ are bounded, based on (\ref{xxx1}) and (\ref{xxx2}), we have $|\mathcal{P}_1(\epsilon) - \mathcal{P}_2| =\mathcal{O}(\epsilon) $.	
\end{proof}

Treating $F_{\mathbf{Q}}(\theta, V)$ as $\mathcal{P}_1(\nu)$ and $F^\dagger_{\mathbf{Q}}(\theta, V)$ as $\mathcal{P}_2$,  since $F_{\mathbf{Q}}(\theta, V, \Omega(\mathbf{Q}))$ and $G_{\mathbf{Q}}(V,\Omega(\mathbf{Q})) $ are smooth and bounded functions,  using the results from Lemma \ref{applemma}, we have $|F_{\mathbf{Q}}(\theta, V)-F^\dagger_{\mathbf{Q}}(\theta, V)|=\mathcal{O}(\nu)=o(1)$, where $o(1)$ asymptotically goes to zero as $\tau$ goes to zero. Since $F^\dagger_{\mathbf{Q}}(\theta, V)=0$  by the definition in (\ref{defapp}), we have $|F_{\mathbf{Q}}(\theta, V)|=o(1)$. 
\end{proof}

Finally, we establish the following lemma to prove the final result.
\begin{Lemma}		\label{tenlemma}
	Suppose $F_{\mathbf{Q}}(\theta^\ast, V^\ast) = 0$ for all $\mathbf{Q}$ together with the transversality condition in (\ref{transodts})  has a unique solution $(\theta^*, V^\ast)$. If $(\theta, V)$ satisfies the approximate Bellman equation in (\ref{conperr}) and the transversality condition in (\ref{transodts}), then $\theta=\theta^\ast+o\left(1 \right)$, $V\left(\mathbf{Q} \right)=V^\ast\left(\mathbf{Q} \right)+o\left(1 \right)$ for all  $\mathbf{Q} $, where the error term $o(1)$ asymptotically goes to zero  as $\tau$ goes to zero.~\hfill\IEEEQED	
\end{Lemma} 
\begin{proof}	[Proof of Lemma \ref{tenlemma}]
	Suppose for some $\mathbf{Q}'$, we have $V\left(\mathbf{Q}' \right)=V^\ast\left(\mathbf{Q}' \right)+\mathcal{O}\left(1 \right)$. From Claim \ref{claimmms}, we have $|F_\mathbf{Q}(\theta, V)|= o(1)$ for all $\mathbf{Q}$. Now let $\tau \rightarrow 0$, we have $(\theta, V)$ satisfies $F_\mathbf{Q}(\theta, V) = 0$ for all $\mathbf{Q}$ and the  transversality condition in (\ref{transodts}). However, $V\left(\mathbf{Q}' \right) \neq V^\ast\left(\mathbf{Q}' \right)$ because of the assumption that $V\left(\mathbf{Q}' \right)=V^\ast\left(\mathbf{Q}' \right)+\mathcal{O}\left(1 \right)$. This contradicts with the condition that $(\theta^*, V^\ast)$ is a unique solution of $F_{\mathbf{Q}}(\theta^\ast, V^\ast) = 0$ for all $\mathbf{Q}$  and the transversality condition in (\ref{transodts}). Hence, we must have $V\left(\mathbf{Q} \right)=V^\ast\left(\mathbf{Q} \right)+o\left(1 \right)$ for all  $\mathbf{Q} $, where $o(1)$ asymptotically goes to zero as $\tau$ goes to zero. Similarly, we can establish $\theta= \theta^\ast + o(1)$.
\end{proof}

\section*{{Appendix C: Proof of Theorem \ref{them111}}}	
For simplicity of notation, we write $J\left(\mathbf{Q}\right)$ in place of $J\left(\mathbf{Q};\boldsymbol{\epsilon}\right)$. We first establish the relationship between $J\left(\mathbf{Q}\right)$ and  $V \left(\mathbf{Q} \right)$. We can observe that if ($c^{\infty}, \{ J\left(\mathbf{Q} \right) \}$) satisfies  the HJB equation in (\ref{cenHJB}), it also satisfies the approximate Bellman equation in (\ref{conperr}). Furthermore, since $J\left(\mathbf{Q}\right)=\mathcal{O}(\sum_{k=1}^K Q_k^2 \log Q_k)$, we have $\lim_{t \rightarrow \infty}\mathbb{E}^{\Omega}\left[J\left(\mathbf{Q}(t)\right) \right]< \infty$ for any admissible policy $\Omega$. Hence, $J\left(\mathbf{Q} \right)=\mathcal{O}(\sum_{k=1}^K Q_k^2 \log Q_k)$ satisfies the transversality condition in (\ref{transodts}). 

Next, we show that the optimal control policy $\Omega^{v \ast}$ obtained by solving the HJB equation in (\ref{cenHJB}) is an admissible control policy in the discrete time system as defined in Definition \ref{admissibledis}. 

Define a \emph{Lyapunov function} as $L(\mathbf{Q})=J \left(\mathbf{Q}\right)$. We further define the \emph{conditional queue drift} as $\Delta (\mathbf{Q})= \mathbb{E}^{\Omega^{v \ast}}\big[\sum_{k=1}^K\left(Q_k(t+1)-Q_k(t) \right)\big|\mathbf{Q}(t)=\mathbf{Q} \big]$ and \emph{conditional Lyapunov drift} as $\Delta L(\mathbf{Q})= \mathbb{E}^{\Omega^{v \ast}}\big[L(\mathbf{Q}(t+1))-L(\mathbf{Q}(t))\big|\mathbf{Q}(t) =\mathbf{Q}\big]$.  We have the following lemma on the relationship between $\Delta (\mathbf{Q})$ and $\Delta L(\mathbf{Q})$.
\begin{Lemma} 	[Relationship between $\Delta (\mathbf{Q})$ and $\Delta L(\mathbf{Q})$]	\label{deldel2}
	$\Delta (\mathbf{Q}) \leq \Delta L(\mathbf{Q})$ if at least one of $\{Q_k: \forall k\}$ is sufficiently large.~\hfill\IEEEQED
\end{Lemma}
\begin{proof}	[Proof of Lemma \ref{deldel2}]	
According to the definition of $\Delta L(\mathbf{Q}) $, we have
\begin{small}	\begin{align}
		\Delta L(\mathbf{Q}) &= \mathbb{E}^{\Omega^{v \ast}}\big[L(\mathbf{Q}(t+1))-L(\mathbf{Q}(t))\big|\mathbf{Q}(t)=\mathbf{Q} \big]	\notag \\
		&\geq\mathbb{E}^{\Omega^{v \ast}}\left[  \sum_{k=1}^K \frac{\partial L(\mathbf{Q})}{\partial Q_k}\left(Q_k(t+1)-Q_k(t) \right)	\bigg|\mathbf{Q}(t)=\mathbf{Q}  \right]	\notag \\
		& \overset{(a)}\geq  \mathbb{E}^{\Omega^{v \ast}}\left[  \sum_{k=1}^K\left(Q_k(t+1)-Q_k(t) \right)\bigg|\mathbf{Q}(t)=\mathbf{Q}  \right]	= \Delta (\mathbf{Q})\quad 	 \label{res2222usew}	
	\end{align}\end{small}if at least one of $\{Q_k: \forall k\}$ is sufficiently large, where $(a)$ is due to the condition that $\left\{\frac{\partial J\left(\mathbf{Q}; \boldsymbol{\epsilon} \right)}{\partial Q_k}: \forall k \right\}$  are increasing functions of all $Q_k$.
\end{proof}

Since $(\lambda_1, \dots, \lambda_K)$ is strictly interior to the stability region $\Lambda$, we have  $(\lambda_1+\delta_1, \dots, \lambda_K+\delta_K)\in \Lambda$  for some positive $\boldsymbol{\delta}=\{\delta_k: \forall k\}$ \cite{capacityregion}. From \emph{Corollary 1} of \cite{neelyitp}, there exists a stationary randomized CSI  only policy $\widetilde{\Omega}$ (that chooses beamforming vectors independent of QSI) such that
\begin{align}
	&\sum_{k=1}^K \mathbb{E}^{\widetilde{\Omega}}\left[ \|\mathbf{w}_k\|^2 \big|\mathbf{Q}(t)=\mathbf{Q}  \right]=\overline{P}(\boldsymbol{\delta}) 	\notag \\
	&\mathbb{E}^{\widetilde{\Omega}}\left[ G_k(\mathbf{H}, \mathbf{w})\big|\mathbf{Q}(t)=\mathbf{Q}  \right] \geq \lambda_k+\delta_k, \quad \forall k		\label{49ersimpor}
\end{align}
where $\overline{P}(\boldsymbol{\delta})$ is the minimum average power required to stabilize the system when arrival rate is $(\lambda_1+\delta_1, \dots, \lambda_K+\delta_K)$. The Lyapunov drift $\Delta L(\mathbf{Q})$ is given by:
\begin{small}\begin{align}
	& \Delta L(\mathbf{Q})+ \mathbb{E}^{\Omega^{v \ast}} \left[ \sum_{k=1}^K  \|\mathbf{w}_k\|^2 \tau  \bigg|\mathbf{Q}(t)=\mathbf{Q}  \right] \notag \\
	\approx  & \sum_{k=1}^K \frac{\partial L(\mathbf{Q})}{\partial Q_k} \lambda_k \tau  + \mathbb{E}^{\Omega^{v \ast}}\left[  \sum_{k=1}^K  \left(\|\mathbf{w}_k\|^2  \tau \right.\right. \notag \\
	& \left.\left.  -\frac{\partial L(\mathbf{Q})}{\partial Q_k}G_k(\mathbf{H}, \mathbf{w})\tau\right)\bigg|\mathbf{Q}(t)=\mathbf{Q}  \right] \notag \\
	\overset{(b)}\leq  & \sum_{k=1}^K \frac{\partial L(\mathbf{Q})}{\partial Q_k} \lambda_k \tau  + \mathbb{E}^{\widetilde{\Omega}}\left[  \sum_{k=1}^K  \left(\|\mathbf{w}_k\|^2  \tau \right.\right. \notag \\
	& \left.\left.	-\frac{\partial L(\mathbf{Q})}{\partial Q_k}G_k(\mathbf{H}, \mathbf{w})\tau\right)\bigg|\mathbf{Q}(t)=\mathbf{Q}  \right]\notag \\
	\overset{(c)} \leq  & -\sum_{k=1}^K \frac{\partial L(\mathbf{Q})}{\partial Q_k} \delta_k \tau  +\overline{P}(\boldsymbol{\delta}) \tau  \notag \\
	\Rightarrow & \Delta L(\mathbf{Q}) \leq  -\sum_{k=1}^K \frac{\partial L(\mathbf{Q})}{\partial Q_k} \delta \tau  +\overline{P}(\boldsymbol{\delta}) \tau	\label{asdadaappne}
\end{align}\end{small}if at least one of $\{Q_k: \forall k\}$ is sufficiently large, where $(b)$ is due to $\Omega^{v \ast}$ achieves the minimum of the HJB equation in (\ref{cenHJB}), and $(c)$ is due to (\ref{49ersimpor}). Since $\left\{\frac{\partial J\left(\mathbf{Q}; \boldsymbol{\epsilon} \right)}{\partial Q_k}: \forall k \right\}$  are increasing functions of all $Q_k$ and combing (\ref{asdadaappne}) with (\ref{res2222usew}), we have 
\begin{align}
	\Delta (\mathbf{Q})<0	\label{driftsada}
\end{align}
if at least one of $\{Q_k: \forall k\}$ is sufficiently large.

Define the \emph{semi-invariant moment generating function} of $A_k-G_k\big(\mathbf{H},\Omega^{v \ast}(\hat{\mathbf{H}},\mathbf{Q})\big)$ as $\phi_k(r,\mathbf{Q})=\ln \big(\mathbb{E}\big[e^{\left(A_k-G_k(\mathbf{H},\Omega^{v \ast}(\hat{\mathbf{H}},\mathbf{Q})  ) \right)r}\big| \mathbf{Q} \big] \big)$. From (\ref{driftsada}), $\mathbb{E}\big[A_k-G_k(\mathbf{H},\Omega^{v \ast}(\hat{\mathbf{H}},\mathbf{Q}))   \big|\mathbf{Q}\big]< 0$ when $Q_k > \overline{Q}_k$ for some large $\overline{Q}_k$. Hence,  $\phi_k(r,\mathbf{Q})$ will have a unique positive root $r_k^\ast(\mathbf{Q})$ ($\phi_k(r_k^\ast(\mathbf{Q}),\mathbf{Q})=0$) \cite{dspgal}.  Let $r_k^\ast= r_k^\ast(\overline{\mathbf{Q}})$, where $\overline{\mathbf{Q}}=(\overline{Q}_1, \dots, \overline{Q}_K)$. We then have the following lemma on the tail distribution,  i.e.,  the complementary cumulative distribution function of $Q_k$, $\Pr\big[ Q_k\geq x\big]$.
\begin{Lemma}	[Kingman Bound \cite{dspgal}]	\label{kingmanres}
	$F_k(x) \triangleq \Pr\big[ Q_k \geq x \big] \leq e^{-r_k^\ast x} $, if $x \geq \overline{x}_k$ for sufficiently large $\overline{x}_k$.~\hfill\IEEEQED
\end{Lemma}

Finally, we check whether $\Omega^{v \ast}$ stabilizes the system according to the definition of the admissible control policy in Definition \ref{admissibledis} as follows:
\begin{small}\begin{align}
	 &\mathbb{E}^{\Omega^{v \ast}} \left[J\left(\mathbf{Q}\right) \right] < \sum_{k=1}^K \mathbb{E}^{\Omega^{v \ast}} \left[ Q_k^3  \right]= \sum_{k=1}^K \left[\int_0^{\infty} \Pr \left[Q_k^3 >s \right] \mathrm{d}s \right] 	\notag \\
	 \leq & \sum_{k=1}^K \left[ \int_{0}^{\overline{x}_k^3} F_k(s^{1/3}) \mathrm{d}s  + \int_{\overline{x}_k^3}^{\infty} F_k(s^{1/3})\mathrm{d}s \right] 	\notag \\
	 \leq & \sum_{k=1}^K \left[\overline{x}_k^3+ \int_{\overline{x}_k^3}^{\infty}  e^{-r_k^\ast s^{1/3}}  \mathrm{d}s \right]	  <  \infty
\end{align}\end{small}
Therefore, $\Omega^{v \ast}$ is an admissible control policy and we have  $V \left(\mathbf{Q} \right)=J\left(\mathbf{Q}\right)$ and $\theta=c^\infty$. Furthermore, using Corollary \ref{cor1}, we have $V^\ast\left(\mathbf{Q} \right)=J\left(\mathbf{Q}\right)+o(1)$ and $\theta^\ast=c^\infty+o(1)$ for sufficiently small $\tau$.

\section*{Appendix D: Proof of Lemma \ref{linearAp}}	
For sufficiently large delay price $\boldsymbol{\gamma}$, ZF beamforming is optimal for  Problem \ref{fluid problem1} when $\boldsymbol{\epsilon}=\mathbf{0}$, i.e., ${\mathbf{h}}_j \mathbf{w}_k^{\ast} = 0$,  $\forall j \neq k$. Therefore, the HJB equation for the base decoupled VCTS when $q_k \geq 0$  ($\forall k$) is given by
\begin{align}	
		& \min_{\{\mathbf{w}: \forall \mathbf{H} \}} \mathbb{E}\bigg[ \sum_{k=1}^K \bigg( \|\mathbf{w}_k\|^2+\gamma_k \frac{q_k}{\lambda_k} -{c_k^{\infty}}  + \frac{\partial  J \left(\mathbf{q};\mathbf{0}\right)} {\partial q_k} \notag \\
		&\Big(- R_k\mathbf{1}\big(R_k\leq \log\big(1+ \big|\mathbf{h}_k  \mathbf{w}_k\big|^2\big)\big)  + \lambda_k\Big) \bigg)\bigg| \mathbf{q} \bigg] =0	\label{orihjbapp11}
\end{align}
Suppose $ J \left(\mathbf{q};\mathbf{0}\right)=\sum_{k=1}^K J_k\left(q_k \right)$, where $J_k\left(q_k \right)$ is the per-flow fluid value function, which is the solution of the following per-flow HJB equation:
\begin{align}		\label{per_flow}
		&\min_{\{\mathbf{w}_k: \forall \mathbf{H} \}} \mathbb{E}\bigg[  \|\mathbf{w}_k\|^2+\gamma_k \frac{q_k}{\lambda_k}-{c_k^{\infty}}   + J_k'\left(q_k \right) \notag \\
		&\Big(- R_k\mathbf{1}\big(R_k\leq \log\big(1+ \big|\mathbf{h}_k  \mathbf{w}_k\big|^2\big)\big)  + \lambda_k\Big) \bigg| q_k \bigg] =0	
\end{align}
Then, the L.H.S. of (\ref{orihjbapp11}) becomes: $\text{L.H.S. of (\ref{orihjbapp11})} 	= \min_{\{\mathbf{w}: \forall \mathbf{H} \}} \mathbb{E}\big[ \sum_{k=1}^K \big( \|\mathbf{w}_k\|^2+\gamma_k \frac{q_k}{\lambda_k} -{c_k^{\infty}}  + J_k'\left(q_k \right) \big(- R_k\mathbf{1}\big(R_k\leq \log\big(1+ \big|\mathbf{h}_k  \mathbf{w}_k\big|^2\big)\big)  + \lambda_k\big) \big)\big| q_k \big]=\sum_{k=1}^K\min_{\{\mathbf{w}_k: \forall \mathbf{H} \}} \mathbb{E}\big[  \|\mathbf{w}_k\|^2+\gamma_k \frac{q_k}{\lambda_k} -{c_k^{\infty}}  + J_k'\left(q_k \right)\big(- R_k\mathbf{1}\big(R_k\leq \log\big(1+ \big|\mathbf{h}_k  \mathbf{w}_k\big|^2\big)\big)  + \lambda_k\big) \big| q_k \big] =0$. Therefore, we show that $J \left(\mathbf{q};\mathbf{0}\right)=\sum_{k=1}^K J_k\left(q_k \right)$ is the solution of (\ref{orihjbapp11}).

Next, we calculate $J_k\left(q_k \right)$ by solving the ODE in (\ref{per_flow}). We first write $\mathbf{w}_k=\sqrt{p_k} \widetilde{\mathbf{w}}_k$, where $\widetilde{\mathbf{w}}_k$ has the same direction as $\mathbf{w}_k$ and has unit norm. Then, the ODE in (\ref{per_flow}) can be written as $\min_{\{\mathbf{p}_{k}: \forall \mathbf{H}\}  } \mathbb{E}\big[ p_k+\gamma_k \frac{q_k}{\lambda_k}  -{c_k^{\infty}}  +   J_k'\left(q_k \right)\left(- R_k \mathbf{1}\left(R_k \leq  \log(1+|\mathbf{h}_k \widetilde{\mathbf{w}}_k|^2 p_k)\right) + \lambda_k \right)\big| q_k\big] =0$. The optimal control action that minimize the L.H.S. of the above equation is given by: $p_k^{\ast}=0$ if $|\mathbf{h}_k \widetilde{\mathbf{w}}_k|^2  \leq  	\frac{2^{R_k}-1}{J_k'\left(q_k \right) R_k }$ and $p_k^{\ast}=\frac{2^{R_k}-1}{ |\mathbf{h}_k \widetilde{\mathbf{w}}_k|^2}$ if $|\mathbf{h}_k \widetilde{\mathbf{w}}_k|^2  >  	\frac{2^{R_k}-1}{J_k'\left(q_k \right) R_k }$. Then the per-flow HJB equation can be written as
\begin{align}		
		 &\mathbb{E}\Big[ p_k^{\ast}\Big| q_k\Big]+\gamma_k \frac{q_k}{\lambda_k}  -{c_k^{\infty}}  +   J_k'\left(q_k \right)	\label{transformperf} \\
		 &\left(- R_k \mathbb{E}\Big[ \mathbf{1}\left(R_k \leq  \log(1+|\mathbf{h}_k \widetilde{\mathbf{w}}_k|^2 p_k^{\ast})\right) \Big| q_k\Big]+ \lambda_k \right) =0	\notag
\end{align}
To solve the ODE in (118), we need to calculate the two terms involving the expectation operator. Since $\mathbf{h}_k \sim \mathcal{CN}\left(0,\mathbf{I}\right)$, we have $\mathbf{h}_k \widetilde{\mathbf{w}}_k \sim \mathcal{CN}\left(0,1\right) $ according to the \emph{bi-unitarily invariant property} \cite{RMT}. Then we have $|\mathbf{h}_k \widetilde{\mathbf{w}}_k|^2 \sim \exp(1)$. Then, $\mathbb{E}\big[ p_k^{\ast}\big| q_k\big] = \int_{\frac{2^{R_k}-1}{J_k'\left(q_k \right) R_k }}^{\infty}\frac{2^{R_k}-1}{x} e^{-x} \mathrm{d}x=(2^{R_k}-1 )E_1\big(\frac{2^{R_k}-1}{J_k'\left(q_k \right) R_k } \big)$ and $\mathbb{E}\big[ \mathbf{1}\left(R_k \leq  \log(1+|\mathbf{h}_k \widetilde{\mathbf{w}}_k|^2 p_k^{\ast})\right) \big| q_k\big]=\int_{\frac{2^{R_k}-1}{J_k'\left(q_k \right) R_k }}^{\infty} 1 \cdot e^{-x} \mathrm{d}x = e^{-\frac{2^{R_k}-1}{J_k'\left(q_k \right) R_k }}$, where $E_1(z)  \triangleq \int_1^{\infty} \frac{e^{-tz}}{t}\mathrm{d}t = \int_z^{\infty} \frac{e^{-t}}{t}\mathrm{d}t $ is the exponential integral function. {We next calculate ${c_k^{\infty}}$. Since ${c_k^{\infty}}$ satisfies the sufficient conditions in (1), (2) in Lemma \ref{HJB1}, we have 
\begin{align}
	R_k  e^{-\frac{2^{R_k}-1}{J_k'\left(0\right) R_k }} =\lambda_k,  \quad a_kE_1\left(\frac{a_k}{J_k'\left(0 \right) R_k } \right)={c_k^{\infty}} 
\end{align}
where $a_k=2^{R_k}-1$. Therefore, ${c_k^{\infty}}=a_k E_1\big(\log \frac{R_k }{\lambda_k} \big)$. }

Substituting the two calculation results and ${c_k^{\infty}}$ into (\ref{transformperf}), we have
\begin{align}		\label{transformODEapp}
		& a_kE_1\Big(\frac{a_k}{J_k'\left(q_k \right) R_k } \Big)+\gamma_k \frac{q_k}{\lambda_k} -{c_k^{\infty}}  +   J_k'\left(q_k \right)	\notag \\
		 & \big(- R_k e^{-\frac{a_k}{J_k'\left(q_k \right) R_k }}+ \lambda_k \big) =0
\end{align}	 
According to  Section \emph{0.1.7.3} of \cite{HandODE}, the parametric solution of the ODE in (\ref{transformODEapp}) is given as shown in (\ref{pareJkEg2}) in Lemma \ref{linearAp}.

\section*{Appendix E: Proof of Corollary \ref{PropJk}}	
Firstly, we obtain the highest order term of $J_k\left(q_k \right)$. The  series expansions of the exponential integral function and exponential function are given as: $E_1(x) = -\gamma_{eu} - \log x - \sum_{n=1}^{\infty} \frac{\left(-x \right)^n}{n ! n}$, $e^x =\sum_{n=0}^{\infty} \frac{x^n}{n!}$.  Based on the parametric solution of $q_k(y)$ in (\ref{pareJkEg2}), we have the following asymptotic property of $q_k(y)$: $q_k(y) =\frac{\lambda_k (R_k  -\lambda_k)}{\gamma_k} \mathcal{O}\left( y\right)$. Similarly,  we have the following asymptotic property of $J_k(y)$: $J_k(y) =\frac{\lambda_k (R_k  -\lambda_k)}{2\gamma_k} \mathcal{O}\left( y^2\right)$. The two asymptotic equations imply that there exists constants $C_1$ and $C_1'$ such that $C_1   y \leq q_k(y) \leq C_1' y$ when $y \rightarrow \infty$. Similarly, there exist constants $C_2$ and $C_2'$ such that $C_2 y^2 \leq J_k(y) \leq C_2' y^2$ when $y \rightarrow \infty$.  Combining the above two inequalities, we have $\frac{C_2}{C_1^{' 2}} q_k^2 \leq J_k(q_k) \leq \frac{C_2'}{C_1^2}  q_k^2$. Therefore, we conclude that $J_k \left( q_k \right) = \mathcal{O} \left(q_k^2 \right),  \text{as } q_k \rightarrow \infty$.

Next, we  obtain the coefficient of the highest order term $q_k^2$. Again, using the series expansion of $E_1(x)$, $e^x$ and the asymptotic property of $J_k\left(q_k \right)$, the per-flow HJB equation in (\ref{transformODEapp})  implies
\begin{align}	\label{final111}
	 J_k'\left(q_k \right) = \frac{\gamma_k}{\lambda_k(R_k -\lambda_k)} q_k + o(q_k)
\end{align}
Furthermore, from $J_k \left( q_k \right) = \mathcal{O} \left(q_k^2 \right)$, we have $\overline{C_1} q_k^2  \leq J_k \left( q_k \right) \leq \overline{C_1'} q_k^2  \Rightarrow 2\overline{C_1} q_k  \leq J_k' \left( q_k \right) \leq 2\overline{C_1'} q_k$. Combining with (\ref{final111}) to match the coefficient of the highest order term of $J_k\left(q_k \right)$, we have $J_k\left(q_k \right) = \frac{\gamma_k}{2\lambda_k(R_k -\lambda_k)} q_k^2 + o(q_k^2)$.

\section*{Appendix F: Proof of Theorem \ref{ErrorEg2}}	
Taking the first order Taylor expansion of the L.H.S. of the HJB equation in (\ref{cenHJB}) at $\boldsymbol{\epsilon}=\mathbf{0}$, $\mathbf{w}=\mathbf{w}^{\ast}$ (where $\mathbf{w}^\ast$ is the optimal control actions given  in (\ref{transformperf}) when $\boldsymbol{\epsilon}=\mathbf{0}$) and using  parametric optimization analysis \cite{paraanay},  we have the following result regarding the approximation error:
\begin{small}
\begin{align}
{J\left(\mathbf{q};\boldsymbol{\epsilon} \right) - J\left(\mathbf{q};\mathbf{0}\right)} =\sum_{k=1}^K \sum_{j \neq k} \epsilon_k  \widetilde{J}_{kj}(\mathbf{q})+ \mathcal{O}(\epsilon^2)	\label{tayloee}
\end{align}
\end{small}where $\widetilde{J}_{kj}(\mathbf{q})$ is meant to capture the coupling terms in $J\left(\mathbf{q};\boldsymbol{\epsilon} \right) $  which satisfies the following PDE:
\begin{small}
\begin{align}
	&\sum_{i=1}^K \left(\lambda_i  - R_i   \mathbb{E} \left[ \mathbf{1}\big(R_i \leq C_i^0(p_i^{\ast}) \big) \big| q_i  \right]  \right) \frac{\partial \widetilde{J}_{kj}\left(\mathbf{q}\right) }{\partial q_i} \notag \\
	 &+\frac{J_k'\left(q_k \right)}{\ln 2}  \mathbb{E} \bigg[ p_j^{\ast} \frac{R_k  \eta e^{\eta(R_k-C_k^0(p_k^{\ast})) }}{(1+e^{\eta(R_k-C_k^0(p_k^{\ast}))})^2} \frac{|\mathbf{h}_k \widetilde{\mathbf{w}}_k^{\ast}|^2 p_k^{\ast}}{1+ |\mathbf{h}_k \widetilde{\mathbf{w}}_k^{\ast}|^2 p_k^{\ast}} \bigg| \mathbf{q} \bigg] =\widetilde{c}_{k}^\infty	\notag
\end{align}
\end{small}with boundary condition $\widetilde{J}_{kj}\left(\mathbf{q} \right)\big|_{\begin{small} q_j = 0\end{small}}=0$ or $\widetilde{J}_{kj}\left(\mathbf{q} \right)\big|_{\begin{small} q_k = 0\end{small}}=0$, where we write $\mathbf{w}_k=\sqrt{p_k} \widetilde{\mathbf{w}}_k$, $C_k^0(p_k^{\ast}) \triangleq \log(1+ |\mathbf{h}_k \widetilde{\mathbf{w}}_k^{\ast}|^2 p_k^{\ast} )$ and $p_k^{\ast}$ is given in Appendix E, and $\widetilde{c}_{k}^\infty=\frac{\partial c^\infty \left(\boldsymbol{\epsilon}\right)}{\partial \epsilon_k}$ is constant (where we treat $c^\infty$ in the coupled system as a function of $\boldsymbol{\epsilon}$). Here we use the logistic function $f^{\eta}\left(x,y \right)=\frac{1}{1+e^{\eta (x-y)}}$ as a smooth approximation for the indicator function in $G(\mathbf{H}, \mathbf{w})$ in (\ref{goodput}), where $\eta > 0$ is a parameter.  Except for the partial differential term, the above PDE  only involves $q_k$ and $q_j$. Therefore, we suppose  $J_{kj}(\mathbf{q})$ is a function of $q_k$ and $q_j$.  Note that $\frac{\eta e^{\eta(R_k-C_k^0(p_k^{\ast})) }}{(1+e^{\eta(R_k-C_k^0(p_k^{\ast}))})^2}$ can be approximated by $\frac{\eta e^{\eta(R_k-C_k^0(p_k^{\ast})) }}{(1+e^{\eta(R_k-C_k^0(p_k^{\ast}))})^2}= \frac{\eta}{5} \mathbf{1}\big(|R_k-C_k^0(p_k^{\ast})|\leq \frac{2}{\eta} \big)$. Without loss of generality, we choose $\eta=5$ and calculate the expectation in the above PDE as follows: $\mathbb{E} \Big[ \frac{ \eta e^{\eta(R_k-C_k^0(p_k^{\ast})) }}{(1+e^{\eta(R_k-C_k^0(p_k^{\ast}))})^2} \cdot\frac{|\mathbf{h}_k \widetilde{\mathbf{w}}_k^{\ast}|^2 p_k^{\ast}}{1+ |\mathbf{h}_k \widetilde{\mathbf{w}}_k^{\ast}|^2 p_k^{\ast}} \Big| \mathbf{q} \Big] =\int_{\frac{2^{R_k}-1}{J_k'\left(q_k \right) R_k }}^{\infty} \frac{2^{R_k}-1}{2^{R_k}} e^{-x}\mathrm{d}x = \frac{2^{R_k}-1}{2^{R_k}} e^{-\frac{2^{R_k}-1}{J_k'\left(q_k \right) R_k }} = \frac{2^{R_k}-1}{2^{R_k}} \mathcal{O}(1)$. Furthermore, we can calculate the other terms involving expectation in the above PDE as follows: $\mathbb{E} \left[ \mathbf{1}\big(R_i \leq C_i^0(p_i^{\ast}) \big) \big| q_i  \right] =e^{-\frac{2^{R_i}-1}{J_i'\left(q_i \right) R_i }}= \mathcal{O}(1)$ and $\mathbb{E} [p_j^{\ast}]=\left(2^{R_j}-1 \right)E_1\left(\frac{2^{R_j}-1}{J_j'\left(q_j \right) R_j } \right) =  \left(2^{R_j}-1 \right)   \mathcal{O}\big(\log(J_j'(q_j))\big)= 2\left(2^{R_j}-1 \right)  \mathcal{O}(\log q_j)$. Substituting the three calculation results  into the above PDE, we obtain $\sum_{i=1}^K \left(\lambda_i  - R_i  \mathcal{O}\left(1\right) \right)  \frac{\partial \widetilde{J}_{kj}\left(\mathbf{q}\right) }{\partial q_i}+D_{kj}'\mathcal{O}\left(q_k \log q_j\right) =\widetilde{c}_{k}^\infty$, where $D_{kj}' \triangleq \frac{\gamma_k(2^{R_j}-1)(2^{R_k}-1)}{\lambda_k (R_k  - \lambda_k) 2^{R_k-1}\ln 2}$. According to Section \emph{3.8.1.2}  of \cite{handbookPDE} and taking into account the boundary condition, we have the leading order terms that $\widetilde{J}_{kj}\left(\mathbf{q}\right)=\frac{D_{kj}'}{R_j-\lambda_j }\mathcal{O} \big(q_kq_j\log q_j \big)$. Substituting it to (\ref{tayloee}), we obtain the approximation error in Theorem \ref{ErrorEg2}.

\section*{Appendix G: Proof of Lemma \ref{equperccp}}
Using the approximate value function in (\ref{perflowapp}), the problem in  (\ref{conperr}) is equivalent to the following per-realization problem: \begin{small}$\underset{\mathbf{w} }\min    \sum_{k=1}^K \left(\|\mathbf{w}_k\|^2 -    \frac{\partial \widetilde {V}\left(\mathbf{Q}\right)}{\partial Q_k}  R_k     (1-\Pr\big[R_k  > C_k\left(\mathbf{H}, \mathbf{w} \right) \big| \hat{\mathbf{H}}, \mathbf{Q}\big]) \right)$\end{small}, where $\frac{\partial \widetilde {V}\left(\mathbf{Q}\right)}{\partial Q_k} $ can be calculated based on (\ref{perflowapp}). Introducing an auxiliary variable $\rho_k = \Pr\big[R_k  > C_k\left(\mathbf{H}, \mathbf{w} \right) \big| \hat{\mathbf{H}}, \mathbf{Q}\big] $, the above is equivalent to the chance constrained problem in Lemma \ref{equperccp}.

\section*{Appendix H: Proof of Lemma \ref{asycvx}}	
In order to verify the convexity of Problem \ref{relaxed}, we just need to verify the convexity of (\ref{linearp4}).  We  write the constraint in (\ref{linearp4}) in the following form: \begin{small}$f\left(\mathbf{W}, \delta_k, x_k, y_k \right) \triangleq  e_k\big(\mathbf{W}  \big) - \mathrm{Tr} \left( \mathbf{M}_k\big(\mathbf{W} \big) \right) +\sqrt{2 \delta_k} x_k + \delta_k y_k \leq 0 $. Since $e_k\big(\mathbf{W}  \big)$\end{small} and $\mathbf{M}_k\big(\mathbf{W}\big)$ are linear (i.e., convex) in $\mathbf{W}$, we have $f\left(\mathbf{W}, \delta_k, x_k, y_k \right)$ is also linear in $\mathbf{W}$. The Hessian matrix of $f\left(\mathbf{W}, \delta_k, x_k, y_k \right)$ is given by: $\mathbf{H}(f) = 
	\begin{small}\left( \begin{array}{cccc}
		\mathbf{H}_{\mathbf{W}} &  & \mathbf{0} \\
		& -\frac{\sqrt{2}}{4}\delta_k^{-\frac{3}{2}}x_k & \frac{\sqrt{2}}{2} \delta_k^{-\frac{1}{2}}& 1 \\
		\mathbf{0} & \frac{\sqrt{2}}{2} \delta_k^{-\frac{1}{2}} & 0 & 0	\\
		& 1 & 0 & 0	
	\end{array} \right) \end{small}$, where $\mathbf{H}_{\mathbf{W}}$ is the Hessian matrix of $f\left(\mathbf{W}, \delta_k, x_k, y_k \right)$ w.r.t. $\mathbf{W}$. Denote $\text{vec}(\mathbf{W})=\left( \text{vec}(W_1), \cdots, \text{vec}(W_K) \right)$ to be the vectorized $\mathbf{W}$, where $\text{vec}(W_k) = \left(W_{k1}^T, \cdots, W_{k N_t}^T \right)$ ($W_{ki}$ is the $i$-th column of $W_k$).  For a given vector $z \triangleq \left(\text{vec}(\mathbf{W}), \delta_k, x_k, y_k \right)$, we calculate the following equation: \begin{small}$z \mathbf{H}(f) z^T = \text{vec}(\mathbf{W}) \mathbf{H}_{\mathbf{W}}\text{vec}^T(\mathbf{W}) + \frac{3 \sqrt{2}}{4} \sqrt{\delta_k} x_k + 2 \delta_k y_k$\end{small}.  Since $f\left(\mathbf{W}, \delta_k, x_k, y_k \right) $ is convex in $\mathbf{W}$, we have $\text{vec}(\mathbf{W}) \mathbf{H}_{\mathbf{W}}\text{vec}^T(\mathbf{W}) \geq 0$. Furthermore, according to the contraints in (\ref{secconp4})--(\ref{psdp42}), we have $\delta_k \geq 0$, $x_k \geq 0$,  $\delta_k \geq 0$, and hence $\frac{3 \sqrt{2}}{4} \sqrt{\delta_k} x_k + 2 \delta_k y_k \geq 0$. Finally, we have \begin{small}$z \mathbf{H}(f) z^T= \text{vec}(\mathbf{W}) \mathbf{H}_{\mathbf{W}}\text{vec}^T(\mathbf{W}) + \frac{3 \sqrt{2}}{4} \sqrt{\delta_k} x_k + 2 \delta_k y_k \geq 0$\end{small}. Therefore, we conclude that Problem \ref{relaxed} is  convex.


\begin{thebibliography}{1}


\bibitem{MIMOBC2}
G. Caire and S. Shamai (Shitz), ``On the achievable throughput of a multiantenna Gaussian downlink system," \emph{IEEE Trans. Inf. Theory}, vol. 49, no. 7, pp. 1691--1706, Jul. 2003.

\bibitem{mmse}
G. Caire, ``MIMO downlink joint processing and scheduling: A survey of classical and recent results,'' in \emph{Proc. Workshop Inform. Theory Applicat.,} UCSD, Feb. 2006.

\bibitem{precoder1}
J. Zhang, Y. Wu, S. Zhou, and J. Wang, ``Joint linear transmitter and receiver design for the downlink of multiuser MU-MIMO systems," \emph{IEEE Commun. Lett.}, vol. 9, pp. 991--993, Nov. 2005.

\bibitem{sumgood2}
W. -C. Li, T. -H. Chang, C. Lin, and C. -Y. Chi, ``A convex approximation approach to weighted sum rate maximization of multiuser MISO interference channel under outage constraints," \emph{in Proc. IEEE ICASSP}, pp. 3368Ð-3371, Progue, Czech, May 2011.

\bibitem{mmse2}
J. Wang and D. P. Palomar, ``Robust MMSE precoding in MIMO channels with pre-fixed receivers,"   \emph{IEEE Trans. Signal Process.},  vol. 58, no. 11, Nov. 2010.

\bibitem{Robustprob1}
A. Mutapcic, S.-J. Kim, and S. Boyd, ``A tractable method for robust downlink beamforming in wireless communications," in \emph{Proc.
Asilomar 2007}, Pacific Grove, CA, Nov. 2007.

\bibitem{luo}
Z.-Q. Luo, W.-K. Ma, A. M.-C. So, Y. Ye, and S. Zhang, ``Semidefinite relaxation of quadratic optimization problems," \emph{IEEE Signal Process. Mag., Special Issue on Convex Optimization for Signal Processing}, pp. 20--34, May. 2010.

\bibitem{sdr2}
C. Ling,  X. Zhang, and L. Qi,  ``Semidefinite relaxation approximation for multivariate bi-quadratic optimization with quadratic constraints,"  \emph{Wiley Online Library: Numerical Linear Algebra with Applications}, vol. 19, no. 1,  pp. 113--131, Dec. 2012.

\bibitem{major2}
E. A. Jorsweick and H. Boche, ``Optimal transmission strategies and impact of correlation in multi-antenna systems with different types of channel state information,"  \emph{IEEE Trans. Signal Process.}, vol. 52, no.12, pp. 3440--3453, Dec. 2004.

\bibitem{surveydelay}
Y. Cui, V. K. N. Lau, R. Wang, H. Huang, and S. Zhang, ``A survey on delay-aware resource control for wireless systems - large deviation theory, stochastic Lyapunov drift and distributed stochastic Learning," \emph{IEEE Trans. Inf. Theory}, vol. 58, no. 3, pp. 1677Ð1701, Mar. 2012.

\bibitem{Cao}
X. Cao, \emph{Stochastic Learning and Optimization: A Sensitivity-Based Approach}. \ Springer, 2008.

\bibitem{DP_Bertsekas}
D. P. Bertsekas, \emph{Dynamic Programming and Optimal Control}, 3rd ed. \ Massachusetts:  Athena Scientific, 2007.

\bibitem{downmimo1}
\tcb{Y. Cui, Q. Huang, and V. K. N. Lau, ``Queue-aware dynamic clustering and power allocation for network MIMO systems via distributed stochastic learning,"  \emph{IEEE Trans. Signal Process}, vol. 59, no. 3, pp. 1229--1238, Mar. 2011.}

\bibitem{downmimo2}
\tcb{Y. Cui, V. K. N. Lau, and Y. Wu, ``Delay-aware BS discontinuous transmission control and user scheduling for energy harvesting downlink coordinated MIMO systems,"  \emph{IEEE Trans. Signal Process}, vol. 60, no. 7, pp. 3786--3795, Jul. 2012.}

\bibitem{chancesss}
A. Nemirovski and A. Shapiro, ``Convex approximations of chance constrained programs," \emph{SIAM J. Optim.}, 17 (2006), pp. 969--996.

\bibitem{Chancecon1}
W. W.-L. Li, Y. J. Zhang, A. M.-C. So, and M. Z. Win,  ``Slow adaptive OFDMA systems through chance constrained programming,"  \emph{IEEE Trans. Signal Process.}, vol. 58, no. 7, pp. 3858--3869, Jul. 2010.

\bibitem{pomdp1}
N. Meuleau, K. E. Kim, L. P. Kaelbling, and A. R. Cassandra, ``Solving pomdps by searching the space of finite policies," in \emph{Proc.
of the Fifteenth Conf. on Uncertainty in AI}, pp. 417--426. 1999.

\bibitem{bernsteininequ}
I. Bechar, ``A Bernstein-type inequality for stochastic processes of quadratic forms of Gaussian variables," available online: http://arxiv.org/abs/0909.3595.

\bibitem{fixedper}
K.-Y. Wang, T.-H. Chang, C.-Y. C. W.-K. Ma, and A. So, ``Probabilistic sinr constrained robust transmit beamforming: A Bernstein-type inequality based conservative approach," \emph{in Proc. IEEE ICASSP}, pp. 3080--3083, 2011.

\bibitem{noisevar}
\tcb{P. Kyritsi, R. Valenzuela, and D. Cox, ``Channel and capacity estimation errors,"  \emph{IEEE Comm. Letters}, vol. 6, no. 12, pp. 517--519, Dec. 2002.}

\bibitem{capacityregion}
M. J. Neely, ``Dynamic power allocation and routing for satellite and wireless networks with time varying channels," \emph{PhD thesis}, Massachusetts Institute of Technology, LIDS, 2003.

\bibitem{equiformu}
\tcb{R. A. Berry and R. G. Gallager, ``Communication over fading channels with delay constraints,"  \emph{IEEE Trans. Inf. Theory,} vol. 48, no. 5, pp. 1135--1149, May 2002.}

\bibitem{boyd}
S. Boyd, \emph{Convex Optimization}. \ Cambridge University Press, 2004.

\bibitem{RB_scheme}
M. Sharif and B.Hassibi, ÒOn the capacity of MIMO broadcast channels with partial side information," \emph{IEEE Trans. Inf. Theory}, vol. 51, no. 2, pp. 506--522, 2005.

\bibitem{neelyitp}
M. J. Neely, ``Energy optimal control for time varying wireless networks," \emph{IEEE Trans. Inf. Theory}, vol. 52, no. 7, pp. 1--18, Jul. 2006.

\bibitem{dspgal}
R. Gallager, \emph{Discrete Stochastic Processes}. \ Boston, MA: Kluwer Academic, 1996.

\bibitem{RMT}
A. M. Tulino and S. Verdœ, \emph{Random Matrix Theory and Wireless Communications}, \ Foundations and Trends in Communication and Information Theory, 2004. 

\bibitem{HandODE}
A. D. Polyanin, V. F. Zaitsev, and A. Moussiaux, \emph{Handbook of Exact Solutions for Ordinary Differential Equations}, 2nd ed.  \ Chapman  \& Hall/CRC Press, Boca Raton, 2003.

\bibitem{paraanay}
J. F. Bonnans and A. Shapiro, ``Optimization problems with perturbations: A guided tour,"  \emph{SIAM Reviews}, vol. 40, no. 2, pp. 228--264, June 1998.

\bibitem{handbookPDE}
A. D. Polyanin, V. F. Zaitsev, and A. Moussiaux, \emph{Handbook of First Order Partial Differential Equations}, 2nd ed.  \  Taylor \& Francis, 2002.


\end{thebibliography}
\end{document}